\documentclass[11pt,reqno]{article}

\usepackage[text={155mm,240mm},centering]{geometry}
\usepackage{graphicx}
\usepackage{amssymb}
\usepackage{epstopdf}

\usepackage{latexsym}

\usepackage{amsmath,amsthm}

\usepackage[mathscr]{eucal}
\usepackage[all]{xy}
\usepackage{mathrsfs}
\usepackage{authblk}
\usepackage{hyperref}

\newtheorem{theorem}{Theorem}[section]

\newtheorem{lemma}[theorem]{Lemma}

\newtheorem{corollary}[theorem]{Corollary}
\newtheorem{proposition}[theorem]{Proposition}

\theoremstyle{definition}
\newtheorem{example}[theorem]{Example}
\newtheorem{example-proposition}[theorem]{Example-Proposition}
\newtheorem{definition}[theorem]{Definition}
\newtheorem{definition-lemma}[theorem]{Definition-Lemma}
\newtheorem{definition-theorem}[theorem]{Definition-Theorem}
\newtheorem{remark}[theorem]{Remark}

\newtheorem*{ack}{Acknowledgements}
\newtheorem*{convention}{Convention}

\newcommand{\ac}{\textup{!`}}

\numberwithin{equation}{section}

\allowdisplaybreaks

\begin{document}

\title{Poisson cohomology, Koszul duality, and Batalin-Vilkovisky algebras}

\author[1]{Xiaojun Chen\thanks{Email: xjchen@scu.edu.cn}}
\author[2]{Youming Chen\thanks{Email: youmingchen@cqut.edu.cn}}
\author[1]{Farkhod Eshmatov\thanks{Email: olimjon55@hotmail.com}}
\author[3]{Song Yang\thanks{Email: syangmath@tju.edu.cn}}

\renewcommand\Affilfont{\small}

\affil[1]{School of Mathematics, Sichuan University, Chengdu 610064, P. R. China}
\affil[2]{School of Science, Chongqing University of Technology, Chongqing 400054, P. R. China}
\affil[3]{Center for Applied Mathematics, Tianjin University, Tianjin 300072, P. R. China}

\date{}

\maketitle

\begin{abstract}

We study the noncommutative Poincar\'e duality
between the Poisson homology and cohomology of unimodular Poisson
algebras, and show that Kontsevich's deformation quantization as well as
Koszul duality preserve the corresponding Poincar\'e duality.
As a corollary, the Batalin-Vilkovisky algebra structures that naturally arise
in these cases are all isomorphic.



\end{abstract}



\section{Introduction}\label{sect:Intro}

In this paper we study the noncommutative Poincar\'e duality
between the Poisson homology and cohomology of unimodular Poisson
algebras, and show that Kontsevich's deformation quantization as well as
Koszul duality preserve the corresponding Poincar\'e duality.

Let $A=\mathbb R[x_1,\cdots, x_n]$ be the real polynomial algebra in $n$ variables.
A Poisson bivector on $A$, say $\pi$, is called {\it quadratic} if it is in the form
\begin{equation}\label{formula:quadratic_Poisson}
\pi=\sum_{i_i,i_2,j_1,j_2}c_{i_1i_2}^{j_1j_2}x_{i_1}x_{i_2}\frac{\partial}{\partial x_{j_1}}\wedge\frac{\partial}{\partial x_{j_2}},
\quad c_{i_1i_2}^{j_1j_2}\in\mathbb R.
\end{equation}
Several years ago, Shoikhet \cite{Shoikhet} observed
that if $\pi$ is quadratic,
then the Koszul dual algebra $A^!$ of $A$, namely,
the graded symmetric algebra $\mathbf\Lambda(\xi_1,\cdots,\xi_n)$
generated by $n$ elements of degree $-1$,
has a Poisson structure (let us call it the {\it Koszul dual} of $\pi$), given by
\begin{equation}\label{corresp:PP}
\pi^!=\sum_{i_1,i_2,j_1,j_2}c_{i_1i_2}^{j_1j_2}\xi_{j_1}\xi_{j_2}\frac{\partial}{\partial \xi_{i_1}}\wedge\frac{\partial}{\partial \xi_{i_2}},
\end{equation}
and proved that Kontsevich's deformation quantization
preserves this type of Koszul duality.
Shoikhet's result motivates us to study some other properties of a
Poisson algebra under
Koszul duality.

First, the following theorem is
clear from Shoikhet's article,
once we explicitly write down the corresponding
complexes.

\begin{theorem}
\label{thm:firsttheorem}
Let $A=\mathbb R[x_1,\cdots, x_n]$ be a quadratic Poisson algebra.
Denote by $A^{!}$ the Koszul dual Poisson algebra of $A$.
Then there are isomorphisms
\begin{equation}\label{iso:HP}
\mathrm{HP}_\bullet(A)\cong\mathrm{HP}^{-\bullet}(A^!; A^{\ac})\quad
\mbox{and}\quad
\mathrm{HP}^\bullet(A)\cong\mathrm{HP}^{\bullet}(A^!),
\end{equation}
where $A^{\ac}:=\mathrm{Hom}_{\mathbb{R}}(A^{!}, \mathbb{R})$ is the linear dual of $A^{!}$.
\end{theorem}

In the above theorem, $\mathrm{HP}_\bullet(-)$ is the Poisson homology,
$\mathrm{HP}^\bullet(-)$ is the Poisson cohomology,
and $\mathrm{HP}^\bullet(A^!;A^{\ac})$ is the Poisson cohomology of $A^!$ with values
in its dual space.

Historically, the Poisson homology and cohomology were introduced by Koszul \cite{Koszul} and
Lichnerowicz \cite{Lichnerowicz}
respectively. In 1997 Weinstein \cite{Weinstein}
introduced the notion of {\it unimodular} Poisson manifolds,
and two years later Xu \cite{Xu} proved that in this case,
there is a Poincar\'e duality between the Poisson
cohomology and homology of $M$.
A purely algebraic version of Weinstein's notion was later formulated
by Dolgushev in \cite{Dolgushev} (see also \cite{LR07,LWW}),
and in this case we also
have
\begin{equation}\label{iso:PD_XuDolgushev}
\mathrm{HP}^\bullet(A)\cong\mathrm{HP}_{n-\bullet}(A),
\end{equation}
for some $n$ depending on $A$.

For a {\it finite dimensional} algebra such as $A^!$ above,
Zhu, Van Oystaeyen and Zhang
introduced in \cite{ZVOZ} the notion of {\it Frobenius Poisson algebra},
and proved that if they
are {\it unimodular} in some sense (to be recalled below),
then there also exists a version of Poincar\'e duality:
\begin{equation}\label{iso:PD_ZVOZ}
\mathrm{HP}^\bullet(A^!)\cong\mathrm{HP}^{\bullet-n}(A^!; A^{\ac}).
\end{equation}

Combining the above two versions of Poincar\'e duality \eqref{iso:PD_XuDolgushev} and
\eqref{iso:PD_ZVOZ}
as well as Theorem \ref{thm:firsttheorem},
we have the following:

\begin{theorem}
\label{thm:secondtheorem}
Let $A=\mathbb R[x_1,\cdots,x_n]$ be a quadratic Poisson algebra.
Then $(A,\pi)$
is unimodular if and only if its Koszul dual $(A^!, \pi^!)$ is unimodular Frobenius.
In this case, we have
the following commutative diagram:
$$
\xymatrixcolsep{4pc}
\xymatrix{
\mathrm{HP}^\bullet(A)\ar[r]^-{\cong}\ar[d]^{\cong}&\mathrm{HP}_{n-\bullet}(A)\ar[d]^{\cong}\\
\mathrm{HP}^\bullet(A^!)\ar[r]^-{\cong}&\mathrm{HP}^{\bullet-n}(A^!; A^{\ac}).
}
$$
\end{theorem}

The main technique to prove the above theorem is the so-called ``differential calculus",
a notion introduced by Tamarkin and Tsygan in \cite{TT05}. Later, Lambre \cite{Lambre}
used the terminology ``differential calculus with duality"
to study the ``noncommutative Poincar\'e duality" in these cases.

In the above-mentioned two references \cite{Xu,ZVOZ},
the authors also proved that
the Poisson cohomology of a unimodular Poisson algebra
(in both cases)
has a Batalin-Vilkovisky algebra structure.
The Batalin-Vilkovisky structure is a very important algebraic structure that has
appeared in, for example, mathematical physics, Calabi-Yau geometry and string topology.
For unimodular quadratic Poisson algebras, we have the following:

\begin{theorem}
\label{thm:thirdtheorem}
Suppose $A=\mathbb R[x_1,\cdots, x_n]$ is a unimodular quadratic Poisson algebra.
Denote by $A^{!}$ its Koszul dual.
Then
$$
\mathrm{HP}^\bullet(A)\cong\mathrm{HP}^\bullet(A^!)
$$
is an isomorphism
of Batalin-Vilkovisky algebras.
\end{theorem}

The above three theorems
have some analogy to
the case of
Calabi-Yau algebras, which
were introduced by Ginzburg \cite{Ginzburg} in 2006.
Suppose a Calabi-Yau algebra, say $A$, is Koszul, then its Koszul dual,
denoted by $A^!$,
is a symmetric Frobenius algebra.
For these two algebras, we also have a version of Poincar\'e duality,
due to Van den Bergh \cite{VdB98} and Tradler \cite{Tradler} respectively
(compare with \eqref{iso:PD_XuDolgushev} and \eqref{iso:PD_ZVOZ}):
$$
\mathrm{HH}^\bullet(A)\cong\mathrm{HH}_{n-\bullet}(A),\quad
\mbox{and}\quad
\mathrm{HH}^\bullet(A^!)\cong\mathrm{HH}^{\bullet-n}(A^!, A^{\ac}).
$$

In \cite[\S 5.4]{Ginzburg} Ginzburg stated a conjecture, which
he attributed to R. Rouquier, saying that for a Koszul Calabi-Yau algebra,
say $A$,
its Hochschild cohomology is isomorphic to the Hochschild cohomology of its Koszul dual $A^!$
\begin{equation}\label{RouquiersconjectureonKCY}
\mathrm{HH}^\bullet(A)\cong\mathrm{HH}^\bullet(A^!)
\end{equation}
as Batalin-Vilkovisky algebras. This conjecture is recently proved by two authors
of the current paper together with G. Zhou in \cite{CYZ}.
In fact, Theorem \ref{thm:thirdtheorem}
may be viewed as a generalization of Rouquier's conjecture in Poisson geometry,
which has been a folklore for several years.

More than just being an analogy, in \cite[Theorem 3]{Dolgushev}, Dolgushev proved that
for the coordinate ring $A$ of an affine Calabi-Yau Poisson variety,
its deformation quantization in the sense of Kontsevich, say $A_{\hbar}$,
is Calabi-Yau if and only if $A$ is unimodular.
Similarly Felder and Shoikhet (\cite{FS}) and later Willwacher and Calaque
(\cite{WC}) proved that,
for a Frobenius Poisson algebra, its deformation quantization
is again symmetric Frobenius if and only if it is unimodular.
Based on these results,
Dolgushev asked two questions in \cite[\S7]{Dolgushev} (see also \cite{DTT}).
The first question is
whether there exists a relationship between the Poincar\'e duality of the
Poisson (co)homology of $A$
and the Poincar\'e duality
of the Hochschild (co)homology of $A_\hbar$.
The following theorem answers this question in the case of
polynomials (the second half also includes the case of Frobenius algebras):

\begin{theorem}
\label{thm:fifttheorem}
$(1)$ Suppose $A=\mathbb R[x_1,\cdots, x_n]$
is a unimodular Poisson algebra. Let $A_{\hbar}$ be its
deformation quantization.
Then the following diagram
$$
\xymatrixcolsep{4pc}
\xymatrix{
\mathrm{HP}^\bullet(A[\![\hbar]\!])\ar[r]^-{\cong}\ar[d]^{\cong}&\mathrm{HP}_{n-\bullet}(A[\![\hbar]\!])
\ar[d]^{\cong}
\\
\mathrm{HH}^\bullet(A_{\hbar})\ar[r]^-{\cong}&\mathrm{HH}_{n-\bullet}(A_{\hbar})}
$$
commutes.

$(2)$ Similarly,
suppose $A^!=\mathbf\Lambda(\xi_1,\cdots,\xi_n)$
is a unimodular Frobenius Poisson algebra, and let $A^!_{\hbar}$ be its
deformation quantization.
Then the following diagram
$$
\xymatrixcolsep{4pc}
\xymatrix{\mathrm{HP}^\bullet(A^![\![\hbar]\!])\ar[r]^-{\cong}\ar[d]^{\cong}
&\mathrm{HP}^{\bullet-n}(A^![\![\hbar]\!]; A^{\ac}[\![\hbar]\!])\ar[d]^{\cong}\\
\mathrm{HH}^\bullet(A^!_{\hbar})\ar[r]^-{\cong}&\mathrm{HH}^{\bullet-n}(A^!_{\hbar}; A^{\ac}_{\hbar})
}
$$
commutes.
\end{theorem}

In other words, the two versions of Poincar\'e duality, one between
the Poisson cohomology and homology, and the other between the Hochschild cohomology
and homology, are preserved under Kontsevich's deformation quantization.

The second question that Dolgushev asked is whether there is any relationship between
the roles that the unimodularity plays
in the above two types of deformation quantizations.
The following theorem partially answers this question, although both cases that Dolgushev
and Felder-Shoikhet/Willwacher-Calaque considered are more general (i.e., not necessarily Koszul):

\begin{theorem}
\label{thm:fourththeorem}
Suppose $A=\mathbb R[x_1,\cdots, x_n]$ is a
quadratic Poisson algebra.
Denote by $A^!$ the Koszul dual algebra of $A$,
and by $A_{\hbar}$ and $A^!_{\hbar}$ the Kontsevich
deformation quantization of $A$ and $A^!$ respectively.
If $A$ is unimodular (and by Theorem \ref{thm:secondtheorem} $A^!$ is unimodular Frobenius), then
$A_{\hbar}$ is Calabi-Yau and $A^!_{\hbar}$ is symmetric Frobenius,
and
the following diagram
\begin{equation}
\xymatrixcolsep{5pc}
\xymatrix{
\mathrm{HP}^\bullet(A[\![\hbar]\!])\ar[r]^-{\cong}\ar[d]^{\cong}
&\mathrm{HP}^{\bullet}(A^![\![\hbar]\!])\ar[d]^{\cong}\\
\mathrm{HH}^\bullet(A_{\hbar})\ar[r]^-{\cong}&\mathrm{HH}^{\bullet}(A^!_\hbar).
}
\end{equation}
is commutative as Batalin-Vilkovisky algebra isomorphisms,
where $A[\![\hbar]\!]$ and $A^![\![\hbar]\!]$ are equipped with the Poisson bivectors $\hbar\pi$ and
$\hbar\pi^!$ respectively.
\end{theorem}

In other words, the theorem says that,
the unimodularity that appears in the deformation
quantization of Calabi-Yau Poisson algebras and
Frobenius Poisson algebras are related by Koszul duality.
Note that in the theorem,
$A_{\hbar}$ and $A^!_{\hbar}$ are Koszul dual to each other
by Shoikhet \cite{Shoikhet}.

Thus as a corollary, one obtains that if $A=\mathbb R[x_1,\cdots, x_n]$ is a unimodular quadratic Poisson algebra,
then the homology and cohomology groups (Poisson and Hochschild)
in Theorems \ref{thm:fifttheorem}
and \ref{thm:fourththeorem} are all isomorphic.
That is, we have the following commutative diagram
of isomorphisms:
$$
\xymatrixrowsep{1pc}
\xymatrixcolsep{1pc}
\xymatrix{
&\mathrm{HP}^\bullet(A^![\![\hbar]\!])\ar[rr] \ar[dd]
&&\mathrm{HP}^{\bullet-n}(A^![\![\hbar]\!]; A^{\ac}[\![\hbar]\!])\ar[dd]\\
\mathrm{HP}^\bullet(A[\![\hbar]\!])\ar[rr] \ar[dd] \ar[ur]&&\mathrm{HP}_{n-\bullet}(A[\![\hbar]\!])\ar[dd]\ar[ur] &\\
&\mathrm{HH}^\bullet(A^!_{\hbar})\ar[rr]&&\mathrm{HH}^{\bullet-n}(A^!_{\hbar}; A^{\ac}_{\hbar})\\
\mathrm{HH}^\bullet(A_{\hbar})\ar[rr]\ar[ur] &&\mathrm{HH}_{n-\bullet}(A_{\hbar}),\ar[ur]&
}
$$
where the horizontal arrows are the Poincar\'e duality, the vertical arrows are given by deformation quantization,
and the slanted arrows are given by Koszul duality.

The rest of the paper is devoted to the proof of the above theorems.
It is organized as follows:
in \S\ref{sect:Pre_Koszul}
we collect several facts on Koszul algebras, and their application
to quadratic Poisson polynomials; in \S\ref{sect:Poissonhomology} we first recall the definition
of Poisson homology and cohomology, and then prove Theorem \ref{thm:firsttheorem};
in \S\ref{sect:unimodular_Poisson}
 we study unimodular quadratic Poisson algebras and their Koszul dual,
and prove Theorem \ref{thm:secondtheorem};
in \S\ref{sect:Iso_BV}
we prove Theorem \ref{thm:thirdtheorem} by means of the so-called ``differential calculus with duality";
in \S\ref{sect:Calabi_Yau_algebras} we discuss Calabi-Yau algebras, their Koszul duality and the Batalin-Vilkovisky
algebras associated to them;
and at last, in \S\ref{sect:connections} we discuss the
deformation quantization of Poisson algebras and prove Theorems \ref{thm:fifttheorem}
and \ref{thm:fourththeorem}.

\begin{ack}
This work is inspired by several interesting conversations of the authors
with P. Smith, S.-Q. Wang and C. Zhu,
to whom we express our gratitude,
during the Non-commutative Algebraic Geometry Workshop 2014 held at Fudan University.
It is partially supported by NSFC (No. 11271269, 11671281)
and RFDP (No. 20120181120090).
\end{ack}

\begin{convention}
Throughout the paper, $k$ is a
field of characteristic zero, which we may assume to be $\mathbb R$ as in \S\ref{sect:Intro}.
All tensors and morphisms are graded over $k$
unless otherwise specified.
For a chain complex, its homology is denoted by $\mathrm{H}_\bullet(-)$, and its cohomology is
$\mathrm{H}^\bullet(-):=\mathrm{H}_{-\bullet}(-)$.
\end{convention}

\section{Preliminaries on Koszul algebras}\label{sect:Pre_Koszul}

In this section, we collect some necessary facts about Koszul algebras.
The interested reader may refer to Loday-Vallette \cite[Chapter 3]{LV} for some more details.


Let $V$ be a finite-dimensional vector space over $k$.
Denote by $TV$ the free (tensor) algebra generated by $V$ over $k$.
Suppose $R$ is a subspace of $V\otimes V$, and let
$(R)$ be the two-sided ideal generated by $R$ in $TV$,
then the quotient algebra
$A:= TV/(R)$
is called
a {\it quadratic algebra}.

Consider the subspace
$$U=\bigoplus_{n=0}^\infty U_n:=
\bigoplus_{n=0}^\infty \bigcap_{i+j+2=n}V^{\otimes i}\otimes R\otimes
V^{\otimes j}$$
of $TV$,
then $U$ is a coalgebra whose coproduct is induced from
the de-concatenation of the tensor products.
The {\it Koszul dual coalgebra} of $A$, denoted
by $A^{\ac}$, is
$$
A^{\ac}=\bigoplus_{n=0}^\infty \Sigma^{\otimes n} (U_n),
$$
where $\Sigma$ is the degree shifting-up (suspension) functor.
$A^{\ac}$ has a graded coalgebra structure induced from that of $U$ with
$$
(A^{\ac})_0=k, \quad (A^{\ac})_1=\Sigma V, \quad (A^{\ac})_2=(\Sigma\otimes\Sigma)(R),\quad\cdots\cdots
$$

The {\it Koszul dual algebra} of $A$, denoted by $A^!$,
is just the linear dual space of $A^{\ac}$, which is then a graded algebra.
More precisely,
let $V^*=\mathrm{Hom}(V, k)$ be the linear dual space of $V$,
and let $R^\perp$ denote
the space of annihilators of $R$ in $V^*\otimes V^*$.
Shift the grading of $V^*$ down by one, denoted by $\Sigma^{-1}V^*$,
then
$$
A^!=T(\Sigma^{-1}V^*)/(\Sigma^{-1}\otimes\Sigma^{-1}\circ R^{\perp}).
\footnote{
In the literature such as \cite{LV}, $A^!$ is defined to be $T(V^*)/R^\perp$, or equivalently,
$(A^!)_i\cong\Sigma^{i}\mathrm{Hom}((A^{\ac})_i, k)$ but not $\mathrm{Hom}((A^{\ac})_i, k)$.
This will cause some issues in our later calculations,
so in this paper, we take $A^!$ as given above, or equivalently $A^!=\mathrm{Hom}(A^{\ac}, k)$.}
$$

Choose a set of basis $\{e_i\}$ for $V$, and let $\{e_i^*\}$ be their duals in $V^*$.
There is a chain complex associated to $A$, called the {\it Koszul complex}:
\begin{equation}\label{Koszul_complex}
\xymatrix{
\cdots\ar[r]^-{\delta}&
A\otimes A^{\ac}_{i+1}\ar[r]^-{\delta}&
A\otimes A^{\ac}_{i}\ar[r]^-{\delta}&
\cdots\ar[r]&
A\otimes A^{\ac}_0\ar[r]^-{\delta}& k,
}
\end{equation}
where for any $r\otimes f\in A\otimes A^{\ac}$,
$\delta(r\otimes f)=\displaystyle\sum_i e_ir\otimes\Sigma^{-1}e_i^*f$.

\begin{definition}[Koszul algebra]
A quadratic algebra $A=TV/(R)$ is called {\it Koszul}
if the Koszul chain complex \eqref{Koszul_complex} is acyclic.
\end{definition}

\begin{example}[Polynomials]\label{Ex:polynomial}
Let $A=k[x_1, x_2,\cdots, x_n]$ be the space of polynomials (the symmetric tensor algebra)
with $n$ generators.
Then
$A$ is a Koszul algebra, and its Koszul dual algebra $A^!$ is
the graded symmetric algebra $\mathbf \Lambda(\xi_1,\xi_2,\cdots,\xi_n)$, with grading
$|\xi_i|=-1$.
\end{example}


\begin{lemma}[Shoikhet \cite{Shoikhet}]\label{thm:Shoikhet}
Let $A=k[x_1,\cdots, x_n]$ with a bivector $\pi$ in the form \eqref{formula:quadratic_Poisson}.
Then
$(A, \pi)$ is quadratic Poisson if and only if
$
(A^!,\pi^!)
$
is quadratic Poisson, where $\pi^!$ is given by \eqref{corresp:PP}.
\end{lemma}


So far, we have assumed that $V$ is a $k$-linear space.
In \S\ref{sect:connections}, we will study the deformed algebras,
which are algebras over $k[\![\hbar]\!]$.
In \cite{Shoikhet}, Shoikhet proved that the definitions and results in above subsections
remain to hold for algebras over a discrete evaluation ring, such as $k[\![\hbar]\!]$.
For example, $k[x_1, \cdots, x_n][\![\hbar]\!]$ is Koszul dual to
$\mathbf\Lambda(\xi_1,\cdots,\xi_n)[\![\hbar]\!]$ as graded algebras over $k[\![\hbar]\!]$ (see \cite[Theorem 0.3]{Shoikhet}).

\section{Poisson homology and cohomology}\label{sect:Poissonhomology}

The notions of Poisson homology and cohomology were introduced by Koszul \cite{Koszul}
and Lichnerowicz \cite{Lichnerowicz} respectively.
Later Huebschmann \cite{Hue90} studied both of them from a purely algebraic perspective.

For a commutative algebra $A$, in the following we denote by $\Omega^p(A)$ the
set of $p$-th K\"ahler differential forms of $A$,
and by
$
\mathfrak X^{-p}(A; M)
$ 
the space of skew-symmetric multilinear maps
$A^{\otimes p}\to M$ that are derivations in each argument. In the following,
if $M=A$, we write $
\mathfrak X^{-p}(A; M)
$ simply by
$
\mathfrak X^{-p}(A)
$.
Note that from the universal property of K\"ahler differentials,
there is an identity
of $A$-modules
\begin{equation}\label{id:Kahlerdiff}
\mathfrak X^{-p}(A; M)=\mathrm{Hom}_A(\Omega^p(A), M).
\end{equation}

\begin{definition}[Koszul \cite{Koszul}]
Suppose $(A,\pi)$ is a Poisson algebra.
Then the {\it Poisson chain complex} of $A$, denoted by $\mathrm{CP}_\bullet(A)$,
is
\begin{equation}
\xymatrix{
\cdots\ar[r]&\Omega^{p+1}(A)\ar[r]^{\partial}&\Omega^{p}(A)\ar[r]^{\partial}
&\Omega^{p-1}(A)\ar[r]^{\partial}&\cdots\ar[r]&\Omega^0(A)=A,
}
\end{equation}
where $\partial$ is given by
\begin{eqnarray*}
\partial(f_0 df_1\wedge \cdots\wedge df_p)&=&\sum_{i=1}^p(-1)^{i-1}\{f_0, f_i\} df_1\wedge\cdots \widehat{df_i}\cdots \wedge df_p\\
&+&\sum_{1\le i<j\le p}(-1)^{j-i}f_0 d\{f_i,f_j\}\wedge df_1\wedge\cdots \widehat{df_i}\cdots\widehat{df_j}\cdots \wedge df_p.
\end{eqnarray*}
The associated homology is called the {\it Poisson homology} of $A$, and is denoted by
$\mathrm{HP}_\bullet(A)$.
\end{definition}

\begin{definition}[Lichnerowicz \cite{Lichnerowicz}]
Suppose $(A,\pi)$ is a Poisson algebra and $M$ is a left Poisson $A$-module.
The {\it Poisson cochain complex} of $A$ with values in $M$, denoted by $\mathrm{CP}^\bullet(A; M)$,
is the cochain complex
$$
\xymatrix{
M=\mathfrak X^0(A; M)\ar[r]^-{\delta}&\cdots\ar[r]&
\mathfrak X^{-p}(A; M)\ar[r]^-{\delta}&\mathfrak X^{-p-1}(A; M)\ar[r]^-{\delta}&\cdots
}$$
where
$\delta$ is given by
\begin{eqnarray*}
\delta(P)(f_0, f_1,\cdots, f_p)&:=&\sum_{0\le i\le p}(-1)^i\{f_i, P(f_0,\cdots, \widehat{f_i},\cdots, f_p)\}\\
&+&\sum_{0\le i<j\le p}(-1)^{i+j}P(\{f_i, f_j\}, f_1, \cdots, \widehat{f_i},\cdots, \widehat{f_j},\cdots, f_p).
\end{eqnarray*}
The associated cohomology is called the {\it Poisson cohomology} of $A$ with values in $M$, and is denoted
by $\mathrm{HP}^\bullet(A; M)$.
In particular, if $M=A$, then the cochain complex is denoted by
$\mathrm{CP}^\bullet(A)$, and the cohomology is 
called the {\it Poisson cohomology} of $A$, and is
denoted by $\mathrm{HP}^\bullet(A)$.
\end{definition}

Note that in the above definition, the Poisson cochain complex, viewed as a chain complex,
is negatively graded,
and the coboundary $\delta$ has degree $-1$. However, by our convention,
the Poisson cohomology are positively graded.

\begin{remark}[The graded case]
The Poisson homology and cohomology can be defined
for graded Poisson algebras as well.
In this case,
$$
\Omega^p(A)=\bigoplus_{n\in\mathbb Z}\Big\{
f_0df_1\wedge \cdots\wedge df_n\Big|f_i\in A,
|f_0|+|f_1|+\cdots+|f_n|+n=p
\Big\}
$$
and
$\mathfrak X^{-p}(A; M)$ is again given by
$\mathrm{Hom}_{A}(\Omega^p(A), M)$.
The boundary maps are completely analogous to
those of Poisson chain and cochain complexes (with Koszul's sign convention taken into account).
\end{remark}

\begin{proof}[Proof of Theorem \ref{thm:firsttheorem}]

(1) We first show the first isomorphism in \eqref{iso:HP}.
Since $A=k[x_1,\cdots,x_n]$, we have an explicit expression
for $\Omega^\bullet(A)$, which is
\begin{equation}\label{formula:Poissonchaincpx}
\Omega^\bullet(A)=\mathbf\Lambda(x_1,\cdots, x_n, dx_1,\cdots, dx_n),
\end{equation}
where $\mathbf\Lambda$ means the graded symmetric tensor product, and
$|x_i|=0$ and $|dx_i|=1$, for $i=1,\cdots, n$.
Similarly,
$$
\Omega^\bullet(A^!)=\mathbf\Lambda(\xi_1,\cdots,\xi_n,d\xi_1,\cdots, d\xi_n),
$$
where
$|\xi_i|=-1$ and $|d\xi_i|=0$ for $i=1,\cdots, n$,
and therefore
\begin{eqnarray}
\mathfrak X^\bullet(A^!; A^{\ac})&=&\mathrm{Hom}_{A^!}(\Omega^\bullet(A^!), A^{\ac})\nonumber\\
&=&\mathrm{Hom}_{\mathbf\Lambda(\xi_1,\cdots,\xi_n)}(\mathbf\Lambda(\xi_1,\cdots,\xi_n,d\xi_1,\cdots, d\xi_n),
\mathrm{Hom}(\mathbf\Lambda(\xi_1,\cdots,\xi_n),k))\nonumber\\
&=&\mathrm{Hom}_{\mathbf\Lambda(\xi_1,\cdots,\xi_n)}(\mathbf\Lambda(\xi_1,\cdots,\xi_n)\otimes\mathbf\Lambda
(d\xi_1,\cdots, d\xi_n),
\mathrm{Hom}(\mathbf\Lambda(\xi_1,\cdots,\xi_n),k))\nonumber\\
&=&\mathrm{Hom}(\mathbf\Lambda
(d\xi_1,\cdots, d\xi_n),
\mathrm{Hom}(\mathbf\Lambda(\xi_1,\cdots,\xi_n),k))\nonumber\\
&=&\mathrm{Hom}(\mathbf\Lambda
(d\xi_1,\cdots, d\xi_n)\otimes\mathbf\Lambda(\xi_1,\cdots,\xi_n), k)\nonumber\\
&=&\mathrm{Hom}(\mathbf\Lambda
(d\xi_1,\cdots, d\xi_n,\xi_1,\cdots,\xi_n), k)\nonumber\\
&=&\mathbf\Lambda\Big(
\frac{\partial}{\partial \xi_1},\cdots,
\frac{\partial}{\partial\xi_n},
\xi_1^*,\cdots,\xi_n^*\Big).\label{formula:coPoissoncochaincpx}
\end{eqnarray}
Thus from \eqref{formula:Poissonchaincpx} and \eqref{formula:coPoissoncochaincpx}
there is a canonical grading preserving isomorphism of vector spaces:
\begin{equation}\label{identificationofPoissonchaincomplexes}
\begin{array}{cccll}
\Phi:&\Omega^\bullet(A)&\longrightarrow&\mathfrak X^{\bullet}(A^!; A^{\ac})&\\
&x_i&\longmapsto&\frac{\partial}{\partial \xi_i}&\\
&dx_i&\longmapsto&\xi_i^*,&i=1,\cdots, n.
\end{array}
\end{equation}
It is a direct check that $\Phi$ is a chain map, and
thus we obtain an isomorphism of Poisson complexes
\begin{equation}\label{iso:Poisson_chain}
\Phi: \mathrm{CP}_\bullet(A)\cong\mathrm{CP}^{-\bullet}(A^!; A^{\ac}),
\end{equation}
which then induces an isomorphism on the homology.

(2) We now show the second isomorphism in \eqref{iso:HP}.
Similarly to the above argument, we have
\begin{eqnarray}\label{Poissoncochainforpolynomials}
\mathrm{CP}^\bullet(A)&=&\mathrm{Hom}_A(\Omega^\bullet(A),A)\nonumber\\
&=&\mathrm{Hom}_{\mathbf\Lambda(x_1,\cdots,x_n)}(\mathbf\Lambda(x_1,\cdots,x_n,dx_1,\cdots,dx_n), \mathbf\Lambda(x_1,\cdots,x_n))
\nonumber\\
&=&\mathrm{Hom}_{\mathbf\Lambda(x_1,\cdots,x_n)}(\mathbf\Lambda(x_1,\cdots,x_n)\otimes
\mathbf\Lambda(dx_1,\cdots,dx_n), \mathbf\Lambda(x_1,\cdots,x_n))\nonumber\\
&=&\mathrm{Hom}(\mathbf\Lambda(dx_1,\cdots,dx_n), \mathbf\Lambda(x_1,\cdots,x_n))\nonumber\\
&=&\mathbf\Lambda\Big(\frac{\partial}{\partial x_1},\cdots,\frac{\partial}{\partial x_n}\Big)\otimes\mathbf\Lambda(x_1,\cdots,x_n)
\end{eqnarray}
and
\begin{eqnarray}\label{Poissoncochainforexteriorpolynomials}
\mathrm{CP}^\bullet(A^!)&=&\mathrm{Hom}_{A^!}(\Omega^\bullet(A^!),A^!)\nonumber\\
&=&\mathrm{Hom}_{\mathbf\Lambda(\xi_1,\cdots,\xi_n)}
(\mathbf\Lambda(\xi_1,\cdots,\xi_n,d\xi_1,\cdots,d\xi_n), \mathbf\Lambda(\xi_1,\cdots,\xi_n))\nonumber\\
&=&\mathrm{Hom}_{\mathbf\Lambda(\xi_1,\cdots,\xi_n)}
(\mathbf\Lambda(\xi_1,\cdots,\xi_n)\otimes\mathbf\Lambda
(d\xi_1,\cdots,d\xi_n), \mathbf\Lambda(\xi_1,\cdots,\xi_n))\nonumber\\
&=&\mathrm{Hom}(\mathbf\Lambda
(d\xi_1,\cdots,d\xi_n), \mathbf\Lambda(\xi_1,\cdots,\xi_n))\nonumber\\
&=&\mathbf\Lambda\Big(\frac{\partial}{\partial \xi_1},\cdots,\frac{\partial}{\partial \xi_n}\Big)
\otimes\mathbf\Lambda(\xi_1,\cdots,\xi_n).
\end{eqnarray}
Under the identity
\begin{equation}\label{mapsthatgiveqis}
x_i\mapsto \frac{\partial}{\partial \xi_i},\quad \frac{\partial}{\partial x_i}\mapsto \xi_i
\end{equation}
we again obtain an isomorphism of chain complexes
$$
\Psi:\mathrm{CP}^\bullet(A)\cong\mathrm{CP}^\bullet(A^!).
$$
This completes the proof.
\end{proof}

\section{Unimodular Poisson algebras and Koszul duality}\label{sect:unimodular_Poisson}

In this section, we study {\it unimodular} Poisson algebras.
We are particularly interested in the algebraic structures
on their Poisson cohomology and homology groups, which are summarized
by {\it differential calculus}, a notion introduced by
Tamarkin and Tsygan in \cite{TT05}.

\begin{definition}[Differential calculus; Tamarkin-Tsygan \cite{TT05}]\label{def:diffcalculus}
Let $\mathrm{H}^{\bullet}$ and $\mathrm{H}_{\bullet}$ be graded vector spaces.
A \textit{differential calculus} is the sextuple
$$
(\mathrm{H}^{\bullet},\mathrm{H}_{\bullet}, \cup, \iota, [-,-], d)
$$
satisfying the following conditions:
\begin{enumerate}
  \item[(1)] $(\mathrm{H}^{\bullet},\cup,[-,-])$ is a Gerstenhaber algebra; that is,
  $(\mathrm H^\bullet, \cup)$ is a graded commutative algebra,
$(\mathrm H^\bullet,[-,-])$ is a degree $1$ or $-1$ graded Lie algebra,
and
the product and Lie bracket are compatible in the following sense
$$
[P\cup Q, R]=P\cup [Q,R]+(-1)^{pq}Q\cup [P,R],
$$
for homogeneous $P, Q, R\in V$ of degree $p,q,r$, respectively;

  \item[(2)] $\mathrm{H}_{\bullet}$ is a graded (left) module over $(\mathrm{H}^{\bullet}, \cup)$ via the map
              $$
               \iota: \mathrm{H}^{n}\otimes\mathrm{H}_{m}\to \mathrm{H}_{m-n},\;f\otimes \alpha\mapsto \iota_f\alpha,
              $$
                for any $f\in \mathrm{H}^{n}$ and $\alpha \in \mathrm{H}_{m}$;
  \item[(3)] There is a map $d: \mathrm{H}_{\bullet}\to \mathrm{H}_{\bullet+1}$
  satisfying $d^2=0$, and moreover, if we set
  ${L}_f:=[d, \iota_f]=d\iota_f-(-1)^{|f|}\iota_f d$,
  then
               $$
              (-1)^{|f|+1}\iota_{[f,g]}=[L_f,\iota_g] :={L}_{f}\iota_g-(-1)^{|g|(|f|+1)}\iota_g {L}_f.
               $$
\end{enumerate}
\end{definition}

In the following, if $\cup$, $\iota$, $[-,-]$ and $d$
are clear from the context, we will simply write
a differential calculus by $(\mathrm H^\bullet,\mathrm H_\bullet)$ for short.

\subsection{Differential calculus on Poisson (co)homology}

Suppose $A$ is a commutative algebra.
Besides the de Rham differential on $\Omega^\bullet(A)$,
we have the following operations on $\mathfrak{X}^\bullet(A)$ and $\Omega^\bullet(A)$:
\begin{enumerate}
\item[(1)] Wedge (cup) product:
suppose $P\in\mathfrak X^{-p}(A)$ and $Q\in\mathfrak X^{-q}(A)$,
then the {\it wedge product} of $P$ and $Q$, denoted by $P\cup Q$, is a polyvector
in $\mathfrak X^{-p-q}(A)$ defined by
$$
(P\cup Q)(f_1,f_2,\cdots, f_{p+q}):=
\sum_{\sigma\in S_{p,q}}\mathrm{sgn}(\sigma)P(f_{\sigma(1)},\cdots, f_{\sigma(p)})
\cdot
Q(f_{\sigma(p+1)},\cdots, f_{\sigma(p+q)}),
$$
where $\sigma$ runs over all $(p,q)$-shuffles of $(1,2,\cdots, p+q)$.

\item[(2)] Schouten bracket:
suppose $P\in\mathfrak X^{-p}(A)$ and $Q\in\mathfrak X^{-q}(A)$,
then their {\it Schouten bracket}, denoted by $[P,Q]$,
is an element in $\mathfrak{X}^{-p-q+1}(A)$ given by
\begin{multline*}
[P, Q](f_1,f_2,\cdots, f_{p+q-1}):=\sum_{\sigma\in S_{q,p-1}}\mathrm{sgn}(\sigma)
P\big(Q(f_{\sigma(1)},\cdots, f_{\sigma(q)}),f_{\sigma(q+1)},\cdots, f_{\sigma(q+p-1)}\big)\\
-(-1)^{(p-1)(q-1)}\sum_{\sigma\in S_{p,q-1}}\mathrm{sgn}(\sigma)
Q\big(P(f_{\sigma(1)},\cdots, f_{\sigma(p)}),f_{\sigma(p+1)},\cdots,f_{\sigma(p+q-1)}\big).
\end{multline*}

\item[(3)] Contraction (inner product): suppose $P\in\mathfrak X^{-p}(A)$ and $\omega=df_1\wedge \cdots\wedge df_n\in\Omega^n(A)$,
then
the {\it contraction} of $P$ with $\omega$, denoted by
$\iota_P(\omega)$, is an $A$-linear map
with values in $\Omega^{n-p}(A)$
given by
$$
\iota_P(\omega)=\left\{
\begin{array}{cl}\displaystyle
\sum_{\sigma\in S_{p,n-p}}\mathrm{sgn}(\sigma)
P(f_{\sigma(1)},\cdots, f_{\sigma(p)})df_{\sigma(p+1)}\wedge\cdots \wedge df_{\sigma(n)},&\mbox{if}\; n\ge p,\\
0,&\mbox{otherwise.}
\end{array}
\right.
$$

\item[(4)] Lie derivative: the {\it Lie derivative} is given by the Cartan formula, namely
for $P\in\mathfrak X^{-p}(A)$ and $\omega\in\Omega^{n}(A)$,
the Lie derivative of $\omega$ with respect to $P$
is given by
$$
L_P\omega:=[\iota_{P},d]=\iota_P(d\omega)-(-1)^p d(\iota_P\omega),
$$
where $d$ is the de Rham differential.
\end{enumerate}

\begin{theorem}\label{DC-on-Poisson}
Suppose $A$ is a Poisson algebra. Then
$$\big(\mathrm{HP}^\bullet(A), \mathrm{HP}_\bullet(A), \cup,\iota, [-,-], d\big),
$$
where $d$ is the de Rham differential, is a differential calculus.
\end{theorem}

\begin{proof}
We only have to show the operations listed above
respect the Poisson boundary and coboundary. It is a direct check and can be found in \cite[Chapter 3]{LGPV}.
\end{proof}


In the following, we will give another differential calculus structure for a Poisson algebra, which will be used later:

(1) For any $P\in\mathfrak X^{-p}(A)$ and $\phi\in\mathfrak X^{-q}(A; A^*)$,
let
$
\iota^*_P(\phi)\in\mathfrak X^{-p-q}(A; A^*)
$
be given by
\begin{equation}\label{iotaofPoisson}
(\iota^*_P\phi)(f_1, \cdots, f_{p+q}):=
\sum_{\sigma\in S_{p,q}}\mathrm{sgn}(\sigma)
P(f_{\sigma(1)},\cdots, f_{\sigma(p)})\cdot
\phi(f_{\sigma(p+1)},\cdots, f_{\sigma(p+q)}).
\end{equation}
It is clear that $\iota^*$ is associative,
i.e., $\iota^*_Q\circ\iota^*_P=\iota^*_{P\cup Q}$.
Also, $\iota^*$ respects the Poisson coboundary maps, which
is completely analogous to the proof of that $\cup$ commutes
with the Poisson coboundary map ({\it cf.} \cite[\S4.3]{LGPV}).

(2) Observe that
\begin{eqnarray}
\mathfrak X^\bullet(A; A^*)&=&\mathrm{Hom}_A(\Omega^\bullet(A), A^*)\nonumber\\
&=&\mathrm{Hom}_A\big(\Omega^\bullet(A), \mathrm{Hom}(A, k)\big)\nonumber\\
&=&\mathrm{Hom}_A(\Omega^\bullet(A)\otimes A, k)\nonumber\\
&=&\mathrm{Hom}(\Omega^\bullet(A), k).\label{vectorfieldswithvalueindual}
\end{eqnarray}
By dualizing the de Rham differential $d$ on $\Omega^\bullet(A)$,
we obtain a differential $d^*$ on $\mathrm{Hom}(\Omega^\bullet(A), k)$,
i.e., on $\mathfrak X^\bullet(A; A^*)$.
It is proved in \cite[Theorem 4.10]{ZVOZ} that
$d^*$ commutes with the Poisson boundary.

(3) For any $P\in\mathfrak X^\bullet(A)$ and $\omega\in\mathfrak X^\bullet(A; A^*)$,
let $L_P\omega:=[\iota^*_P, d^*](\omega)$; it is a direct check that
$$
[L_P, \iota^*_Q]=\iota^*_{[P, Q]}.
$$

By (1)-(3) listed above, we obtain the following.

\begin{theorem}\label{DC-on-Poisson-dual}
Suppose $A$ is a Poisson algebra, and denote $A^*$ be its dual space. Then
$$\big(\mathrm{HP}^\bullet(A),\mathrm{HP}^\bullet(A; A^*),\cup,\iota^*,[-,-], d^*\big)$$
is a differential calculus.
\end{theorem}

We next introduce two DG Lie algebras associated to the above two differential calculi.
Let us start with the notion of negative cyclic homology.

\begin{definition}[Cyclic homology; {\it cf.} Jones \cite{Jones} and Kassel \cite{Kassel}]\label{def_HC}
Suppose $(\mathrm C_\bullet, b, B)$ is a mixed complex, with $|b|=-1$ and $|B|=1$.
Let
$u$ be a free variable of degree $-2$ which commutes with $b$ and $B$.
The {\it negative cyclic chain complex} of $\mathrm C_\bullet$ is the following complex
\begin{eqnarray*}
(\mathrm C_\bullet[\![u]\!], b+uB),
\end{eqnarray*}
and is denoted by $\mathrm{CC}^{-}_\bullet(\mathrm C_\bullet)$.
The associated homology is called the {\it negative cyclic homology} of $\mathrm C_\bullet$, and is denoted by
$\mathrm{HC}_\bullet^{-}(\mathrm C_\bullet)$.
\end{definition}


\begin{remark}[Cyclic cohomology]\label{rmk:cycliccohomology}
Suppose $(\mathrm C^\bullet, b, B)$ is a mixed cochain complex, namely
$|b|=1$ and $|B|=-1$. By negating the degrees of $\mathrm C^\bullet$, we obtain a
mixed chain complex, denoted by $(\mathrm C_\bullet, b, B)$ with $|b|=-1$ and $|B|=1$.
By our convention,
the {\it cyclic cohomology}
of $(\mathrm C^\bullet, b, B)$, denoted by $\mathrm{HC}^\bullet(\mathrm C^\bullet)$, is the {\it cohomology}
of the negative cyclic complex of $(\mathrm C_\bullet, b, B)$.
\end{remark}



Consider the mixed complex $\Omega^\bullet(A)$ with differential $(0, d)$,
where $d$ is the de Rham differential.
Equip $\mathfrak X^\bullet(A)$
with trivial differential.
Since $\Omega^\bullet(A)$ is a Lie module over $\mathfrak X^\bullet(A)$ whose action commutes with $d$,
the negative cyclic complex $(\Omega^\bullet(A)[\![u]\!], ud)$
is a DG module over $\mathfrak X^\bullet(A)$.
Consider the semi-direct product
\begin{equation}\label{DGLAofPoisson}
\mathfrak P(A)^{\#}:=\Sigma\mathfrak X^\bullet(A)\ltimes\Sigma^{-1-n}\Omega^\bullet(A)[\![u]\!],
\end{equation}
where $n$ is an arbitrary integer number.
It is a DG Lie algebra with differential $(0, ud)$.

Similarly, for the mixed complex $(\mathfrak X(A^*), 0, d^*)$, we have the DG Lie algebra
\begin{equation}\label{DGLAofFrobeniusPoisson}
\mathfrak P^{\circ}(A)^{\#}
:=\Sigma\mathfrak X^\bullet(A)\ltimes\Sigma^{-1-n}\mathfrak X^\bullet(A; A^*)[\![u]\!],
\end{equation}
with differential given by $(0, ud^*)$.

\subsection{Unimodular Poisson algebras}

Suppose $A$ is a commutative algebra, and
$\eta\in\Omega^n(A)$.
We say $\eta$ is a volume form if $\mathfrak X^\bullet (A)\stackrel{\iota_{(-)}\eta}{\longrightarrow} \Omega^{n+\bullet}(A)$
is an isomorphism of vector spaces.
Now suppose $A$ is Poisson, then we have the following diagram
\begin{equation}\label{diag:unimodularPoisson}
\xymatrixcolsep{4pc}
\xymatrix{
\mathfrak X^\bullet (A)\ar[r]^-{\iota_{(-)}\eta}&\Omega^{n+\bullet}(A)\\
\mathfrak X^{\bullet+1} (A)\ar[r]^-{\iota_{(-)}\eta}\ar[u]_{\delta}&\Omega^{n+\bullet+1}(A),\ar[u]_{\partial}
}
\end{equation}
which may not be commutative, i.e., $\eta$ may not be a Poisson cycle.
We say $A$ is {\it unimodular} if there exists a volume form $\eta$ such that
\eqref{diag:unimodularPoisson} commutes.

In terms of the DG Lie algebra \eqref{DGLAofPoisson},
being unimodular is equivalent to the following.

\begin{proposition}\label{prop:altofunimodular}
Let $A$, $\pi$ and $\eta$ be as above. Then the bivector $\pi$
is unimodular Poisson
if and only if $(\Sigma\pi, \Sigma^{-1-n}\eta)$ is a Maurer-Cartan element of
the DG Lie algebra \eqref{DGLAofPoisson}.
\end{proposition}

The proof is a direct check, and
we leave it to the interested reader.
Recall that for a DG Lie algebra $(L, d)$,
any Maurer-Cartan element, say $a\in L$,
gives a new DG Lie algebra structure on $L$ with
differential $\tilde d=d+[a, -]$. Denote this DG Lie algebra
by $L_a$. Going back to the above proposition,
in the following we denote
$$
\mathfrak P(A,\eta):=\mathfrak P(A)^{\#}_{(\Sigma\pi, \Sigma^{-1-n}\eta)},
$$
which will be used later in \S\ref{sect:connections}.

The following is also immediate
from \eqref{diag:unimodularPoisson}.

\begin{theorem}[Xu]\label{thm:Xu}
Suppose $A$ is a unimodular Poisson algebra with the volume form
of degree $n$.
Then $(
\mathrm{HP}^\bullet(A),
\mathrm{HP}_{n-\bullet}(A))$
forms a differential calculus with duality, and therefore
there exists an isomorphism (the Poincar\'e duality)
$$
\mathrm{HP}^\bullet(A)\cong
\mathrm{HP}_{n-\bullet}(A).
$$
\end{theorem}

\subsection{Unimodular Frobenius Poisson algebras}

Now, we go to unimodular Frobenius Poisson algebras, a notion
introduced by  Zhu, Van Oystaeyen and Zhang in
\cite{ZVOZ}.

Suppose $A^!$ is a finite dimensional graded not-necessarily commutative algebra.
$A^!$ is called {\it symmetric Frobenius}
if it is equipped with a bilinear, non-degenerate symmetric pairing
$$
\langle-,-\rangle: A^!\otimes A^!\to k
$$
of degree $n$ which is cyclically invariant, that is,
$
\langle a, b\cdot c\rangle=(-1)^{(|a|+|b|)|c|}\langle c, a\cdot b\rangle
$, for all homogeneous $a, b, c\in A^!$.
This is equivalent to saying that there is an $A^!$-bimodule isomorphism
$$
\eta^!:(A^{!})^\bullet\longrightarrow (A^{\ac})_{n+\bullet},\quad\mbox{for some\;} n\in\mathbb N,
$$
where $A^{\ac}=(A^{!})^*$.
In this case, we
may view $\eta^!$ as an element in $\mathrm{Hom}_{A^!}(A^!, A^{\ac})\subset\mathfrak X^\bullet(A^!; A^{\ac})$.
Now assume $A^!$ is Poisson,  then we have a diagram
\begin{equation}\label{formula:unimodularcyclicPoisson}
\xymatrixcolsep{4pc}
\xymatrix{
\mathfrak X^\bullet (A^!) \ar[r]^-{\iota^*_{(-)}\eta^!}&\mathfrak X^{n+\bullet}(A^!; A^{\ac})\\
\mathfrak X^{\bullet+1} (A^!)\ar[r]^-{\iota^*_{(-)}\eta^!}\ar[u]_{\delta}&\mathfrak X^{n+\bullet+1}(A^!; A^{\ac}).\ar[u]_{\delta}
}
\end{equation}
According to
Zhu-Van Oystaeyen-Zhang \cite{ZVOZ},
if there exists $\eta^!\in\mathfrak X^\bullet(A^!; A^{\ac})$
such that $\iota^*_{(-)}\eta^!$ is an isomorphism, then
$\eta^!$ is called a {\it volume form},
and if furthermore,
the digram
\eqref{formula:unimodularcyclicPoisson}
commutes,
then $A^!$ is called a {\it unimodular Frobenius Poisson algebra} of degree $n$
(in \cite{ZVOZ} the authors call it {\it unimodular Fronbenius Poisson}).
From the definition, we immediately have:

\begin{theorem}[Zhu-Van Oystaeyen-Zhang \cite{ZVOZ}]\label{thm:ZVOZ}
Suppose $A^!$ is a unimodular Frobenius Poisson algebra with the volume form
of degree $n$.
Then $(
\mathrm{HP}^\bullet(A^!),
\mathrm{HP}^{\bullet-n}(A^!; A^{\ac}))
$
forms a differential calculus with duality
and therefore there exists an isomorphism
$$
\mathrm{HP}^\bullet(A^!)\cong
\mathrm{HP}^{\bullet-n}(A^!; A^{\ac}).
$$
\end{theorem}

In this paper, since we are interested in $A=k[x_1,\cdots, x_n]$
or $A^!=\mathbf\Lambda(\xi_1,\cdots, \xi_n)$,
we always assume the volume form is constant.
The following is completely analogous to
Proposition \ref{prop:altofunimodular}:

\begin{proposition}\label{altdefofsymmetricunimodularity}
Suppose $A^!=\mathbf\Lambda(\xi_1,\cdots, \xi_n)$ with
volume form $\eta^!$.
Then a bivector $\pi^!$
is unimodular Frobenius Poisson
if and only if $(\Sigma\pi^!, \Sigma^{-1-n}\eta^!)$ is
a Maurer-Cartan element of the DG Lie algebra
$
\mathfrak P^{\circ}(A^!)^{\#}
$ given by \eqref{DGLAofFrobeniusPoisson}.
\end{proposition}

In the following \S\ref{sect:connections}
we shall use the DG Lie algebra
$$
\mathfrak P^{\circ}(A^!,\eta^!)
:=\mathfrak P^{\circ}(A^!)^{\#}_{(\Sigma\pi^!, \Sigma^{-1-n}\eta^!)}.$$

\begin{proof}[Proof of Theorem \ref{thm:secondtheorem}]
First, we show that
a quadratic Poisson algebra
$(A=k[x_1,\cdots,x_n], \pi)$ is unimodular if and only if $(A^!,\pi^{!})$ is unimodular Frobenius.
In fact, recall that for $A=k[x_1,\cdots,x_n]$,
$$
\begin{array}{ll}
\displaystyle\mathfrak{X}^\bullet (A)=\mathbf\Lambda\Big(x_1,\cdots, x_n,
\frac{\partial}{\partial x_1},\cdots, \frac{\partial}{\partial x_n}\Big), &
\displaystyle
\Omega^\bullet(A)=\mathbf\Lambda(x_1,\cdots, x_n, dx_1,\cdots, dx_n),\\
\displaystyle\mathfrak X^\bullet(A^!)=
\mathbf\Lambda\Big(\xi_1,\cdots,\xi_n, \frac{\partial}{\partial \xi_1},\cdots,\frac{\partial}{\partial \xi_n}\Big),
&
\displaystyle\mathfrak X^\bullet(A^!; A^{\ac})=\mathbf\Lambda\Big(\xi_1^*,\cdots, \xi_n^*,
\frac{\partial}{\partial \xi_1},\cdots,\frac{\partial}{\partial \xi_n}\Big).
\end{array}
$$
Let
$$
\eta=dx_1dx_2\cdots dx_n\quad\mbox{and}\quad
\eta^!=\xi_1^*\xi_2^*\cdots \xi_n^*,
$$
where $\eta^!$ is understood as contraction, namely,
$$\eta^!(\xi_{i_1}\cdots \xi_{i_p}):=\sum_{\sigma\in S_{p, n-p}}\langle\xi_{i_1}\cdots \xi_{i_p},
\xi^*_{\sigma(1)}\cdots \xi^*_{\sigma(p)}\rangle\cdot\xi^*_{\sigma(p+1)}\cdots \xi^*_{\sigma(n)},$$
then under the identification
\begin{equation}\label{idoftwodifferentialcalculus}
x_i\mapsto\frac{\partial}{\partial \xi_i},
\quad
dx_i\mapsto\xi_i^*,
\quad
\frac{\partial}{\partial x_i}\mapsto
\xi_i
\end{equation}
the diagram
\begin{equation}\label{cd_fourhomology}
\xymatrixcolsep{2pc}
\xymatrix{
\mathfrak{X}^\bullet (A)=\mathbf\Lambda\left(x_1,\cdots, x_n,
\frac{\partial}{\partial x_1},\cdots, \frac{\partial}{\partial x_n}\right)
\ar[r]^-{\iota_{(-)}\eta}\ar[d]^{\cong}&
\Omega^\bullet(A)=\mathbf\Lambda\left(x_1,\cdots, x_n, dx_1,\cdots, dx_n\right)\ar[d]^{\cong}\\
\mathfrak X^\bullet (A^!)=
\mathbf\Lambda\left(\xi_1,\cdots,\xi_n, \frac{\partial}{\partial \xi_1},\cdots,\frac{\partial}{\partial \xi_n}\right)
\ar[r]^{\iota^*_{(-)}\eta^!}
&
\mathfrak X^\bullet (A^!; A^{\ac})=\mathbf\Lambda\left(\xi_1^*,\cdots, \xi_n^*,
\frac{\partial}{\partial \xi_1},\cdots,\frac{\partial}{\partial \xi_n}\right)
}
\end{equation}
commutes.
This means $\eta$ is a Poisson cycle for $A$ if and only if $\eta^!$ is a Poisson
cocycle for $A^!$, which proves the claim.

Second, for $A$ as above, we show the following diagram
\begin{equation}\label{diag:four_iso}
\xymatrixcolsep{4pc}
\xymatrix{
\mathrm{HP}^\bullet(A)\ar[r]^-{\cong}\ar[d]^{\cong}&\mathrm{HP}_{n-\bullet}(A)\ar[d]^{\cong}\\
\mathrm{HP}^\bullet(A^!)\ar[r]^-{\cong}&\mathrm{HP}^{\bullet-n}(A^!; A^{\ac}).
}
\end{equation}
commutes.
In fact, the two vertical isomorphisms are given by Theorem \ref{thm:firsttheorem},
and the two horizontal isomorphisms are given by Theorems \ref{thm:Xu} and \ref{thm:ZVOZ} respectively.
The commutativity of the diagram \eqref{diag:four_iso} follows from the chain level commutative diagram \eqref{cd_fourhomology}.
\end{proof}

\begin{remark}
By the same identification \eqref{idoftwodifferentialcalculus},
one immediately sees that for quadratic Poisson algebra $A$
and its Koszul dual $A^!$, the two DG Lie algebras
given by \eqref{DGLAofPoisson} and \eqref{DGLAofFrobeniusPoisson}
are isomorphic.
\end{remark}

\section{Poisson cohomology and the Batalin-Vilkovisky algebra}\label{sect:Iso_BV}

The purpose of this section is to show that for unimodular quadratic Poisson polynomial algebras,
the horizontal isomorphisms in \eqref{diag:four_iso}
naturally induce on $\mathrm{HP}^\bullet(A)$ and $\mathrm{HP}^\bullet(A^!)$
a Batalin-Vilkovisky algebra structure, and the vertical isomorphisms in \eqref{diag:four_iso}
are isomorphisms of Batalin-Vilkovisky algebras.
We start with the notion of {\it differential calculus with duality}.

\begin{definition}[Lambre \cite{Lambre}]
A differential calculus $(\mathrm{H}^{\bullet},\mathrm{H}_{\bullet}, \cup, \iota, [-,-], d)$ is
called a \textit{differential calculus with duality} if there exists an integer $n$ and an element $\eta\in \mathrm{H}_{n}$
such that
\begin{enumerate}
\item[(a)]  $\iota_{1} \eta=\eta$, where $1\in \mathrm{H}^{0}$ is the unit, $d(\eta)=0$, and
\item[(b)]  for any $i\in \mathbb Z$,
            \begin{equation}\label{PD}
            \mathrm{PD}(-):=\iota_{(-)}\eta : \mathrm{H}^{i}\to \mathrm{H}_{n-i}
            \end{equation}
            is an isomorphism.
\end{enumerate}
Such isomorphism $\mathrm{PD}$ is called the \textit{Van den Bergh duality} (also called \textit{the
noncommutative Poincar\'e duality}),
and $\eta$ is called the {\it volume form}.
\end{definition}

\begin{definition}[Batalin-Vilkovisky algebra]
Suppose $(V,\bullet )$
is an graded commutative algebra. A {\it Batalin-Vilkovisky
algebra} structure on $V$ is the triple $(V,\bullet , \Delta)$
such that
\begin{enumerate}
\item[$(1)$] $\Delta: V^{i}\to V^{i-1}$ is a differential, that is, $\Delta^2=0$; and
\item[$(2)$] $\Delta$ is second order operator,
              that is,
              \begin{eqnarray*}
                \Delta(a\bullet  b\bullet  c)
                  &=& \Delta(a\bullet  b)\bullet  c+(-1)^{|a|}a\bullet \Delta(b\bullet  c)+(-1)^{(|a|-1)|b|}b\bullet \Delta(a\bullet  c)\nonumber \\
                  & & - (\Delta a)\bullet  b\bullet  c-(-1)^{|a|}a\bullet (\Delta b)\bullet  c-(-1)^{|a|+|b|}a\bullet  b\bullet  (\Delta c).
              \end{eqnarray*}
\end{enumerate}
\end{definition}

Equivalently,
if we define the bracket
$$
[a,b]:=(-1)^{|a|+1}(\Delta(a\bullet  b)-\Delta(a)\bullet  b- (-1)^{|a|}a\bullet  \Delta(b)),
$$
then $[-,-]$ is a derivation with respect to $\bullet $ for each component.
In other words, a  Batalin-Vilkovisky algebra is a Gerstenhaber algebra $(V,\bullet , [-,-])$
with a differential $\Delta: V^{i}\to V^{i-1}$ such that
\begin{equation}\label{G-BV-equ}
[a,b]=(-1)^{|a|+1}(\Delta(a\bullet  b)-\Delta(a)\bullet  b-(-1)^{|a|}a\bullet  \Delta(b)),
\end{equation}
for any $a, b\in V$ (\textit{cf}. \cite[Proposition 1.2]{Getzler}).
$\Delta$ is also called the Batalin-Vilkovisky operator, or the generator (of the Gerstenhaber bracket).

Now suppose $(\mathrm{H}^{\bullet},\mathrm{H}_{\bullet},\cup,\iota,[-,-],d,\eta)$
is a differential calculus with duality. Let
$\Delta: \mathrm{H}^\bullet\to\mathrm{H}^{\bullet-1}$ be the linear operator
such that

\begin{equation}\label{Def-of-BV-operator}
\xymatrixcolsep{4pc}
\xymatrix{
\mathrm H^\bullet\ar[r]^-{\Delta}\ar[d]^{\mathrm{PD}}&\mathrm H^{\bullet-1}\ar[d]^{\mathrm{PD}}\\
\mathrm H_{n-\bullet}\ar[r]^-{d}&\mathrm H_{n-\bullet+1}
}
\end{equation}
commutes.
Then we have the following theorem:

\begin{theorem}[Lambre \cite{Lambre}]\label{Thm_Lambre}
Let $(\mathrm{H}^{\bullet},\mathrm{H}_{\bullet},\cup,\iota,[-,-],d,\eta)$
be a differential calculus with duality.
Then the triple $(\mathrm{H}^{\bullet},\cup,\Delta)$ is a Batalin-Vilkovisky algebra.
\end{theorem}

The proof can be found in Lambre (\cite[Th\'{e}or\`{e}me 1.6]{Lambre}); however, since some details in loc. cit.
are omitted, we give a proof here for completeness.

\begin{proof}
Since $(\mathrm{H}^{\bullet},\cup, [-,-])$ is a Gerstenhaber algebra,
we only need to show that the Gerstenhaber bracket
is compatible with the operator $\Delta$ in \eqref{Def-of-BV-operator}; that is, equation \eqref{G-BV-equ} holds.
For any homogeneous elements $f,g\in \mathrm{H}^{\bullet}$,
by the definition of Poincar\'{e} duality $\mathrm{PD}$ \eqref{PD} and the Cartan formulae (Lemma \ref{Cartan formulas}),
we have
\begin{eqnarray*}
&&(-1)^{|f|+1}\mathrm{PD}([f,g])\\
&=&(-1)^{|f|+1} \iota_{[f,g]}(\eta) =
  [ {L}_f,\iota_g](\eta) =  {L}_f\iota_g(\eta)-(-1)^{|g|(|f|+1)}\iota_g  {L}_f(\eta)\\
   &=& d \iota_f\iota_g(\eta)- (-1)^{|f|}\iota_f d \iota_g(\eta)-(-1)^{|g|(|f|+1)}\iota_g d\iota_f(\eta)
        +(-1)^{|g|(|f|+1)+|f|}\iota_g\iota_f d(\eta)\\
   &=& d \circ \mathrm{PD}(f\cup g)-(-1)^{|g|(|f|+1)}\iota_g d \circ \mathrm{PD}(f)-(-1)^{|f|}\iota_f d\circ \mathrm{PD}(g)\\
   &=& \mathrm{PD}(\Delta(f\cup g))-(-1)^{|g|(|f|+1)}\iota_g \mathrm{PD}(\Delta(f))- (-1)^{|f|}\iota_f \mathrm{PD}(\Delta(g))\\
   &=& \iota_{\Delta(f\cup g)}(\eta)-(-1)^{|g|(|f|+1)}\iota_g \iota_{\Delta(f)} (\eta)- (-1)^{|f|}\iota_f \iota_{\Delta(g)}(\eta))\\
   &=& (\iota_{\Delta(f\cup g)}-(-1)^{|g|(|f|+1)}\iota_{g\cup \Delta(f)}-(-1)^{|f|}\iota_{f\cup \Delta(g)})(\eta)\\
   &=& \mathrm{PD}(\Delta(f\cup g)-\Delta(f)\cup g-(-1)^{|f|}f\cup \Delta(g)).
\end{eqnarray*}
Since $\mathrm{PD}$ is an isomorphism,
we thus have
\begin{equation*}
[f,g]=(-1)^{|f|+1}(\Delta(f\cup g)-\Delta(f)\cup g-(-1)^{|f|}f\cup \Delta(g)).
\qedhere
\end{equation*}
\end{proof}

\begin{corollary}[see also Xu \cite{Xu} and Zhu-Van Oystaeyen-Zhang \cite{ZVOZ}]\label{thm:Batalin-Vilkovisky}
Suppose $A$ is a unimodular Poisson or unimodular Frobenius Poisson algebra.
Then $\mathrm{HP}^\bullet(A)$ admits a Batalin-Vilkovisky algebra structure.
\end{corollary}

\begin{proof}
If $A$ is unimodular Poisson, then Theorems \ref{DC-on-Poisson} and \ref{thm:Xu}
imply the pair $(\mathrm{HP}^\bullet(A), \mathrm{HP}_\bullet(A))$ is in fact a differential calculus with duality;
similarly, if $A$ is unimodular Frobenius Poisson, Theorem \ref{DC-on-Poisson-dual} and \ref{thm:ZVOZ}
$(\mathrm{HP}^\bullet(A), \mathrm{HP}^\bullet(A; A^*))$ is a differential calculus with duality.
The theorem then follows from Theorem \ref{Thm_Lambre}.
\end{proof}

\begin{proof}[Proof of Theorem \ref{thm:thirdtheorem}]

Note that in Theorem \ref{thm:secondtheorem},
the right vertical isomorphism preserves the K\"ahler differential as well as
the volume form,
that is, the two differential calculus with duality
$$
\big(\mathrm{HP}^{\bullet}(A),\mathrm{HP}_{\bullet}(A)\big)\;
\mbox{and} \;\big (\mathrm{HP}^{\bullet}(A^{!}),\mathrm{HP}^{\bullet}(A^{!};A^{\ac})\big)
$$
are isomorphic.
Combining with Corollary \ref{thm:Batalin-Vilkovisky},
the theorem follows.
\end{proof}

\begin{remark}
Not all quadratic Poisson algebras are unimodular.
For example, for $A=\mathbb R[x_1,x_2,x_3]$,
Etingof-Ginzburg \cite[Lemma 4.2.3 and Corollary 4.3.2]{EG10}
showed that any unimodular Poisson structure is of the form
$$
\{x,y\}=\frac{\partial\phi}{\partial z},\quad
\{y,z\}=\frac{\partial\phi}{\partial x},\quad
\{z,x\}=\frac{\partial\phi}{\partial y},$$
for some $\phi\in A$ (taking $\phi$ to be cubic then the Poisson structure is quadratic);
for $A=\mathbb{C}[x_1,x_2,x_3,x_4]$,
Pym \cite[\S3]{Pym} showed that
any unimodular quadratic Poisson bracket on $A$ may be written uniquely in the following form
$$
\{f,g\}:=\frac{df\wedge dg\wedge d\alpha}{dx_1\wedge dx_2\wedge dx_3\wedge dx_4}, \; f, g\in A,
$$
where $\alpha=\sum_{i=1}^{4} \alpha_i dx_i \in \Omega^{1}(A)$ such that $\alpha\wedge d\alpha=0$,
and $\alpha_i$'s are homogeneous cubic polynomials satisfying $\sum_{i=1}^{4}x_i \alpha_i=0$.
\end{remark}

\section{Calabi-Yau algebras}\label{sect:Calabi_Yau_algebras}

At the end of \S\ref{sect:Intro}
we sketched some analogy between
unimodular Poisson algebras
and Calabi-Yau algebras.
In the following two sections,
we study their relationships in more detail.

\subsection{Calabi-Yau algebras and
the Batalin-Vilkovisky algebra structure}\label{sect:CYandcalculus}

\begin{definition}[Calabi-Yau algebra; Ginzburg \cite{Ginzburg}]
Let $A$ be an associative algebra over $k$.
$A$ is called a {\it Calabi-Yau algebra of dimension $n$}
if
\begin{enumerate}
\item[$(1)$] $A$ is homologically smooth, that is, $A$, viewed as an $A^e$-module,
has a bounded resolution of finitely generated projective $A^e$-modules, and
\item[$(2)$] there is an isomorphism
\begin{equation}\label{CY_cond}
\mathrm{RHom}_{A^{e}}(A, A\otimes A)\cong \Sigma^{-n}A
\end{equation}
in the derived category $D(A^e)$ of $A^e$-modules.
\end{enumerate}
\end{definition}

In the above definition, $A^e$ is the enveloping algebra of $A$, namely $A^{e}:=A\otimes A^{\mathrm{op}}$.
There are a lot of examples of Calabi-Yau algebras,
such as the universal enveloping algebra
of semi-simple Lie algebras, the skew-product of complex polynomials
with a finite subgroup of $\mathrm{SL}(n,\mathbb R)$, the Yang-Mills algebras, etc.

We next study Van den Bergh's
noncommutative Poincar\'e duality for Calabi-Yau algebras
(\cite{VdB98}). To this end, we first recall the differential calculus structure for associative algebras.


For a unital associative algebra $A$, let $\bar A=A/k$ be its augmentation, and
$A\to\bar A: a\mapsto\bar a$ be the projection. Denote
by $(\bar{\mathrm C}^\bullet(A),\delta)$ and
$(\bar{\mathrm C}_\bullet(A), b)$
the reduced Hochschild cochain and chain complexes of $A$
(the reader may refer to Loday \cite{Loday}
for notations).
Recall that the {\it Gerstenhaber cup product} and the {\it Gerstenhaber bracket} on
$\bar{\mathrm{C}}^\bullet(A)$
are given as follows: for any $f\in \bar{\mathrm{C}}^{n}(A)$ and $g \in \bar{\mathrm{C}}^{m}(A)$,
$$
f\cup g(\bar{a}_1,\ldots, \bar{a}_{n+m})
:=(-1)^{nm}
f(\bar{a}_1, \ldots, \bar{a}_{n})g(\bar{a}_{n+1},\ldots,\bar{a}_{n+m}),
$$
and
$$
\{f,g\}:=f\circ g-(-1)^{(|f|+1)(|g|+1)}g\circ f,
$$
where
\begin{equation*}
f{\circ} g (\bar{a}_1,\ldots, \bar{a}_{n+m-1})
:=
\sum^{n-1}_{i=0}(-1)^{(|g|+1)i} f(\bar{a}_1, \ldots, \bar{a}_{i},
\overline{g(\bar{a}_{i+1},\ldots, \bar{a}_{i+m})}, \bar{a}_{i+m+1},\ldots, \bar{a}_{n+m-1}).
\end{equation*}

Gerstenhaber proved in \cite[Theorems 3-5]{Gerstenhaber} $\cup$ and $\{-,-\}$
are well-defined
on the cohomology level, and moreover, $\cup$ is graded commutative. Therefore we obtain
on the Hochschild cohomology $\mathrm{HH}^{\bullet}(A)$ a Gerstenhaber algebra structure.

Next, we consider the action of the Hochschild cochain complex on the Hochschild chain complex.
Given any homogeneous elements $f\in \bar{\mathrm{C}}^{n}(A)$
and $\alpha=(a_0,\bar{a}_1, \ldots, \bar{a}_m)\in \bar{\mathrm{C}}_{m}(A)$,
\begin{enumerate}
\item[$(1)$]
the \textit{cap product}
$
\cap: \bar{\mathrm{C}}^{n}(A)\times \bar{\mathrm{C}}_{m}(A)\to \bar{\mathrm{C}}_{m-n}(A)
$
is given by
\begin{equation}\label{def:cap1}
f \cap \alpha  :=\left\{
\begin{array}{cl}
(a_0f(\bar{a}_1,\ldots,\bar{a}_n),\bar{a}_{n+1},\ldots,\bar{a}_m),&\mbox{if}\; m\geq n\\
0,&\mbox{otherwise.}
\end{array}\right.
\end{equation}
If we denote by $\iota_{f}(-):= f\cap -$ the contraction operator,
then $\iota_{f}\iota_{g}=(-1)^{|f||g|}\iota_{g\cup f}=\iota_{f\cup g}$;

\item[$(2)$] the \textit{Lie derivative} $
L: \bar{\mathrm{C}}^{n}(A)\times \bar{\mathrm{C}}_{m}(A)\to \bar{\mathrm{C}}_{m-n}(A)
$ is given as follows:
for any $\alpha=(a_0, \bar{a}_1,\ldots,\bar{a}_m)\in \bar{\mathrm{C}}_{m}(A)$,
if $n\le m+1$, then
\begin{eqnarray*}
L_f(\alpha)
&:=&  \sum^{m-n}_{i=0}(-1)^{(n+1)i}(a_0,\bar{a}_1\cdots,\bar{a}_i, \overline{f(\bar{a}_{i+1},\cdots,\bar{a}_{i+n})}, \cdots, \bar{a}_m)\\
&+& \sum_{i=m-n+1}^{m}(-1)^{m(i+1)+n+1}
    (f(\bar{a}_{i+1}, \cdots, \bar{a}_m,\bar{a}_{0},\ldots,\bar{a}_{n-m+i-1}), \bar{a}_{n-m+i},\ldots, \bar{a}_{i}),
\end{eqnarray*}
where the second sum is taken over all cyclic permutations such that $a_0$ is inside of $f$,
and otherwise if $n>m+1$, $L_{f}(\alpha)=0$;

\item[$(3)$] the \textit{Connes operator}
$
B: \bar{\mathrm{C}}_{\bullet}(A)\to \bar{\mathrm{C}}_{\bullet+1} (A)
$
is given by
$$
B(\alpha)
:=\sum_{i=0}^{m} (-1)^{mi} (1,  \bar{a}_i, \cdots,   \bar{a}_m, \bar{a}_0,\cdots, \bar{a}_{i-1}).
$$
\end{enumerate}


The following two lemmas first appeared in Daletskii-Gelfand-Tsygan \cite{DGT},
which we learned from Tamarkin-Tsygan in \cite{TT05}.

\begin{lemma}\label{cap-Liede-lemma}
Keep the notations as in the above definition. Then
\begin{enumerate}
\item[$(1)$] $(\bar{\mathrm{C}}_{\bullet}(A),b,\cap)$ is a DG module
over $(\bar{\mathrm{C}}^{\bullet}(A), \delta, \cup)$,
that is,
\begin{equation*}
 \iota_{\delta f}=(-1)^{|f|+1}[b,\iota_f],\quad
\iota_f \iota_g=\iota_{f\cup g},
\end{equation*}
for any homogeneous elements $f, g\in \bar{\mathrm{C}}^{\bullet}(A)$;

\item[$(2)$] for any homogeneous elements $f, g\in \bar{\mathrm{C}}^{\bullet}(A)$,
\begin{equation*}
  [L_f,L_g] = L_{\{f,g\}},
\end{equation*}
and in particular $(-1)^{|f|+1} {[b,L_f]}+L_{{\delta f}}= 0$.
\end{enumerate}
\end{lemma}

\begin{lemma}[Homotopy Cartan formulae]\label{Cartan formulas}
Suppose $\iota, L, B$ are given as above and $f, g\in \bar{\mathrm{C}}^{\bullet}(A)$ are any homogeneous elements.
\begin{enumerate}
\item[$(1)$]
Define an operation (cf. \cite[Equ. (3.5)]{TT05})
\begin{eqnarray*}
S_{f}(\alpha)&:=& \sum_{i=0}^{m-n} \sum_{j=i+n}^{m} (-1)^{\eta_{ij}}
 (1, \bar{a}_{j+1},\cdots, \bar{a}_{m},\bar{a}_0,\cdots,\bar{a}_i,
    \overline{f(\bar{a}_{i+1} ,\cdots,\bar{a}_{i+n})},
     \bar{a}_{i+n+1},\cdots, \bar{a}_{j})
\end{eqnarray*}
for any $\alpha=(a_0,\bar{a}_1,\cdots,\bar{a}_m)\in \bar{\mathrm{C}}_{m}(A;A)$
(the sum is taken over all cyclic permutations and $a_0$ always appears on the left of $f$),
where $\eta_{ij}:=(n+1)m+(m-j)m+(n+1)(j-i)$.
Then we have
\begin{equation}\label{homotopy Cartan formula1}
L_{f}=[B,\iota_f]+[b,S_f]-S_{\delta f}.
\end{equation}

\item[$(2)$] Define
\begin{eqnarray*}
T(f,g)(\alpha)&:=&\sum_{i=l-n+2}^l\sum_{j=0}^{n+i-l-2}(-1)^{\theta_{ij}}\\
&&(f(\bar{a}_{i+1},\cdots,\bar{a}_l,\bar{a}_0,\cdots,\bar{a}_j,
\overline{g(\bar{a}_{j+1},\cdots,\bar{a}_{j+m})},\cdots,\bar{a}_{n+m+i-l-2}),\cdots,\bar{a}_i)
\end{eqnarray*}
for any $\alpha=(a_0,\bar{a}_1,\cdots,\bar{a}_l)\in \bar{\mathrm{C}}_{l}(A;A)$,
where $\theta_{ij}=(m+1)(i+j+l)+l(i+1)$.
Then we have
\begin{equation}\label{homotopy Cartan formula2}
[L_{f},\iota_g]-(-1)^{|f|+1}\iota_{\{f,g\}}=[b,T(f,g)]-T(\delta f,g)-T(f,\delta g).
\end{equation}
\end{enumerate}
\end{lemma}

The above two lemmas say that Definition \ref{def:diffcalculus}
(2) (3) hold up homotopy on the chain level. Together with Gerstenhaber's theorem, we have the following.

\begin{theorem}[Daletskii-Gelfand-Tsygan \cite{DGT}]\label{DGT thm}
Let $A$ be an associative algebra.
Then the following sextuple
$$
\big(\mathrm{HH}^{\bullet}(A),\mathrm{HH}_{\bullet}(A), \cup, \iota, \{-,-\}, B\big)
$$
is a differential calculus.\end{theorem}


In \cite[Proposition 5.5]{dTdVVdB}, de Thanhoffer de V\"olcsey and Van den Bergh proved that,
for a Calabi-Yau algebra $A$ of dimension $n$, there exists a class $\eta\in\mathrm{HH}_n(A)$
such that the contraction
\begin{equation}\label{thm:dTdVVdB}
\xymatrixcolsep{2.5pc}
\xymatrix{
\mathrm{HH}^\bullet(A)\ar[r]^-{-\cap \eta}&
\mathrm{HH}_{n-\bullet}(A)
}
\end{equation}
is an isomorphism. This immediately implies the following:

\begin{theorem}[\cite{Ginzburg,Lambre}]\label{PDofVandenBergh}
Suppose $A$ is a Calabi-Yau algebra $A$ of dimension $n$.
Then
$$
(\mathrm{HH}^{\bullet}(A),\mathrm{HH}_{\bullet}(A), \cup, \iota, \{-,-\}, B)
$$
is a differential calculus with duality, and
in particular, $(\mathrm{HH}^\bullet(A),\cup, \Delta)$ is a Batalin-Vilkovisky algebra.
\end{theorem}

\subsection{Symmetric Frobenius algebras and the Batalin-Vilkovisky algebra structure}\label{subsectionofcyclic}

We now recall a differential calculus structure on the Hochschild complexes of symmetric Frobenius algebras.


First, for an associative algebra $A$,
denote $A^*:=\mathrm{Hom}(A, k)$, which is an $A$-bimodule.
Denote by $\bar{\mathrm{C}}^\bullet(A; A^*)$ the reduced Hochschild
cochain complex of $A$ with values in $A^*$.
Then under the identity
\begin{equation}\label{Hochwithvalueindual}
\bar{\mathrm{C}}^\bullet(A; A^*)=\bigoplus_{n\geq 0}\mathrm{Hom}(\bar{A}^{\otimes n}, A^*)
=\bigoplus_{n\geq 0}\mathrm{Hom}(A\otimes \bar{A}^{\otimes n}, k)
=\mathrm{Hom}(\bar{\mathrm{C}}_\bullet(A), k),
\end{equation}
one may equip on $\bar{\mathrm{C}}^\bullet(A; A^*)$
the dual Connes differential, which is denoted by $B^*$, i.e., $B^{*}(g):=(-1)^{|g|}g\circ B$ for homogeneous $g \in \bar{\mathrm{C}}^\bullet(A; A^*)$.
$B^*$ commutes with the Hochschild coboundary map $\delta$, and thus is well-defined on the homology level.

Second, let
\begin{equation}\label{def:cap2}
\begin{array}{ccl}
\bar{\mathrm{C}}^{\bullet}(A)\times
\bar{\mathrm{C}}^{\bullet}(A; A^*)&\stackrel{\cap^{\ast}}{\longrightarrow}&
\bar{\mathrm{C}}^\bullet(A; A^*)\\
(f, \quad\quad \alpha)&\longmapsto&
\iota^{\ast}_{f}(\alpha):=(-1)^{|f||\alpha|}\alpha\circ \iota_f,
\end{array}
\end{equation}
for any homogeneouss $f\in \bar{\mathrm{C}}^{\bullet}(A)$ and $\alpha\in
\bar{\mathrm{C}}^{\bullet}(A; A^*)$.
We have the following.

\begin{theorem}\label{dual-Hoch-DC}
Let $A$ be an associative algebra.
Then
$$
(\mathrm{HH}^\bullet(A),\mathrm{HH}^\bullet(A; A^*),\cup,\iota^{\ast},\{-,-\}, B^*)
$$
is a differential calculus.
\end{theorem}

\begin{proof}By the definition of differential calculus, we only need to show the last two equalities given
in Definition \ref{def:diffcalculus}.

(1)
By the definition of $\iota^{\ast}$ and Lemma \ref{cap-Liede-lemma} (1),
one has
\begin{eqnarray*}
\iota^{\ast}_f \iota^{\ast}_g(\alpha)
    &=& (-1)^{|g||\alpha|}\iota^{\ast}_f(\alpha\circ\iota_g)=(-1)^{|g||\alpha|+|f|(|\alpha|+|g|)}(\alpha\circ \iota_g)\circ\iota_f \\
    &=& (-1)^{|g||\alpha|+|f|(|\alpha|+|g|)} \alpha\circ (\iota_{g\cup f})=(-1)^{|f||g|}\iota^{\ast}_{g\cup f}\alpha
    = \iota^{\ast}_{f\cup g} (\alpha),
\end{eqnarray*}
for any homogenous elements $f,g\in \mathrm{HH}^{\bullet}(A)$ and $\alpha \in\mathrm{HH}^{\bullet}(A; A^*)$.
This means that the cap product is a left module action.

(2)
Given any homogenous elements $f\in \mathrm{HH}^{\bullet}(A)$ and $\alpha \in\mathrm{HH}^{\bullet}(A; A^*)$,
define
\begin{equation}\label{LieactiononHochschildcochain}
L^{\ast}_{f}(\alpha):=(-1)^{|f||\alpha|+|\alpha|+1} \alpha\circ L_f (=[B^{\ast},\iota^{\ast}_{f}](\alpha)),
\end{equation}
and by Lemma \ref{Cartan formulas} one has
\begin{eqnarray*}
  [L^{\ast}_{f},\iota^{\ast}_{g}](\alpha)
   &=& (L^{\ast}_{f}\iota^{\ast}_{g}-(-1)^{(|f|+1)|g|}\iota^{\ast}_{g}L^{\ast}_{f})(\alpha)\\
   &=&(-1)^{(|f|+1)(|\alpha|+|g|)+|g||\alpha|+1}\alpha\circ(\iota_{g}L_{f})-(-1)^{(|f|+|g|+1)|\alpha|+1} \alpha\circ (L_{f}\iota_{g})\\
   &=&(-1)^{(|f|+|g|+1)|\alpha|} \alpha\circ ([L_{f},\iota_{g}])\\
   &=&(-1)^{(|f|+|g|+1)|\alpha|} \alpha\circ ((-1)^{|f|+1}\iota_{\{f,g\}})\\
   &=&(-1)^{|f|+1}\iota^{\ast}_{\{f,g\}}(\alpha).
\end{eqnarray*}
This completes the proof.
\end{proof}


Now suppose $A^!$ is a symmetric Frobenius algebra. Recall that
the existence of the degree $n$ cyclic pairing
is equivalent to an isomorphism
$$
\eta: A^!\cong \Sigma^{-n}A^{\ac}
$$
as $A^!$-bimodules.
Such $\eta$ may be viewed
as an element in
$\bar{\mathrm{C}}^{-n}(A^!; A^{\ac})$,
which is a cocycle, and hence
represents a cohomology class.
By abuse of notation, this class is also
denoted by $\eta$. The following map
\begin{multline}
-\cap^{\ast}\eta: \bar{\mathrm{C}}^{\bullet}(A^!)
=\bigoplus_{q\geq 0}
\mathrm{Hom}((\bar{A^!})^{\otimes q}, A^!)
\\
\stackrel{\eta\circ-}{\longrightarrow}
\bigoplus_{q\geq 0} \mathrm{Hom}((\bar{A^!})^{\otimes q}, \Sigma^{-n}A^{\ac})=\bar{\mathrm{C}}^{\bullet-n}(A^!; A^{\ac}),\label{def:capwithvol}
\end{multline}
where $\eta\circ-$ means composing with $\eta$, gives
an isomorphism on the cohomology (due to Tradler \cite{Tradler}). Thus we have the following.

\begin{theorem}[\cite{Lambre,Tradler}]\label{thm:dualityofTradler}
Suppose $A^!$ is a symmetric Frobenius algebra of degree $n$.
$$
(\mathrm{HH}^\bullet(A),\mathrm{HH}^\bullet(A; A^*),\cup,\iota^{\ast},\{-,-\}, B^*)
$$
is a differential calculus with duality, and in particular, $\mathrm{HH}^\bullet(A^!)$ is a Batalin-Vilkovisky algebra.
\end{theorem}

\begin{remark}\label{Gerstenhabercomodule}
Suppose
$
(\mathrm{H}^\bullet,\mathrm{H}_\bullet,\cup,\iota,\{-,-\}, B)
$
is a differential calculus, then $\iota$ and the Lie derivative 
$L=[\iota, B]$ is nothing but saying that $\mathrm H_\bullet$ is a Gerstenhaber module
over $\mathrm H^\bullet$.
From this point of view, the two differential calculus structures given in Theorems
\ref{thm:ZVOZ}
and \ref{thm:dualityofTradler}
can be understood in the following way:
Since $(\mathrm{HP}^\bullet(A^!), \mathrm{HP}_\bullet(A^!))$ already forms
a differential calculus and $\mathrm{HP}^\bullet(A^!; A^{\ac})$
is the linear dual of $\mathrm{HP}_\bullet(A^!)$ (see \eqref{vectorfieldswithvalueindual}),
the Gerstenhaber module structure on $\mathrm{HP}^\bullet(A^!; A^{\ac})$
is exactly the dual (or say adjoint) of Gerstenhaber module structure
on $\mathrm{HP}_\bullet(A^!)$. 
Analogously, by \eqref{Hochwithvalueindual},
$\mathrm{HH}^\bullet(A^!; A^{\ac})$ is the linear dual of $\mathrm{HH}_\bullet(A^!)$,
and thus the differential calculus structure
on $(\mathrm{HH}^\bullet(A^!), \mathrm{HH}^\bullet(A^!; A^{\ac}))$
can also be understood from this point of view.
\end{remark}

\subsection{Koszul Calabi-Yau algebras and Rouquier's conjecture}
Analogously to the quadratic Poisson algebra case,
the Koszul dual of a Koszul Calabi-Yau algebra is symmetric Frobenius
(chronologically this fact is discovered first), and we have the following
theorem due to Van den Bergh
(see \cite[Theorem 9.2]{VdB97} or \cite[Proposition 28]{CYZ} for a proof):
Suppose $A$ is a Koszul algebra and let $A^{!}$ be its Koszul dual algebra.
Then $A$ is Calabi-Yau of dimension $n$
if and only if $A^{!}$ is symmetric Frobenius of degree $n$.

It has been well-known that for a Koszul algebra, say $A$,
$$
\mathrm{HH}^\bullet(A)\cong\mathrm{HH}^\bullet(A^{!}),
$$
as Gerstenhaber algebras,
and Rouquier conjectured (it is stated in Ginzburg \cite{Ginzburg}) that,
for a Koszul Calabi-Yau algebra,
the above two Batalin-Vilkovisky are isomorphic, which turns out to be true
(see \cite[Theorem A]{CYZ} for a proof):

\begin{theorem}[Rouquier's conjecture]\label{thm:conjofRouquier}
Suppose $A$ is a Koszul Calabi-Yau algebra. Denote by $A^!$ and by $A^{\ac}$ the Koszul
dual algebra and coalgebra of $A$ respectively.
Then
$$
\big(\mathrm{HH}^\bullet(A),\mathrm{HH}_\bullet(A)\big)\quad
\mbox{and}\quad
\big(\mathrm{HH}^\bullet(A^!),\mathrm{HH}^\bullet(A^!; A^{\ac})\big)
$$
are isomorphic as differential calculus with duality. In particular,
$
\mathrm{HH}^\bullet(A)
$
and
$
\mathrm{HH}^\bullet(A^!)
$
are isomorphic as Batalin-Vilkovisky algebras.
\end{theorem}

The key point of the proof is that,
with the differentials properly assigned on $A\otimes A^!$ and $A\otimes A^{\ac}$ respectively, then
$$\bar{\mathrm{C}}^\bullet(A; A)\simeq A\otimes A^! \simeq\bar{\mathrm{C}}^\bullet(A^!; A^!)\quad\mbox{and}\quad
\bar{\mathrm{C}}_\bullet(A; A)\simeq A\otimes A^{\ac} \simeq\bar{\mathrm{C}}^\bullet(A^!; A^{\ac}),$$
and via these quasi-isomorphisms, the volume forms as well as
the contractions given by \eqref{def:cap1} and \eqref{def:cap2}
are identical on the above middle terms (compare with the proof of Theorem \ref{thm:secondtheorem}).

\begin{example}[The polynomial case]\label{exampleofCY}
Let $A=\mathbb R[x_1,x_2,\cdots,x_n]$, which is $n$-Calabi-Yau.
Its Koszul dual algebra $A^!=\mathbf\Lambda(\xi_1,\xi_2,\cdots, \xi_n)$
is symmetric Frobenius. As in the Poisson case,
the volume classes on $\mathrm{HH}_\bullet(A)$ and $\mathrm{HH}^\bullet(A^!; A^{\ac})$
are, via the above quasiisomorphisms, represented
by
$1\otimes\xi_1^*\cdots\xi_n^*$ in $A\otimes A^{\ac}$.
\end{example}

We would like to summarize some results of the previous two
subsections in terms of DG Lie algebras analogous to
the ones given by \eqref{DGLAofPoisson}
and \eqref{DGLAofFrobeniusPoisson}.

For an $n$-Calabi-Yau algebra $A$ with volume form $\eta$,
$(0,\Sigma^{-1-n}\eta)$ is a Maurer-Cartan element of the following
DG Lie algebra of semi-direct product
\begin{equation}\label{DGLAofCalabiYau}
\mathfrak D(A)^{\#}:=\Sigma\bar{\mathrm{C}}^\bullet(A)\ltimes
\Sigma^{-1-n}\overline{\mathrm{CC}}_\bullet^{-}(A).
\end{equation}
Let
$$
\mathfrak D(A,\eta):=\mathfrak D(A)^{\#}_{(0,\Sigma^{-1-n}\eta)},
$$
then it is a DG Lie algebra, and will be studied in the next section.

For a symmetric Frobenius algebra $A^!$ with volume form $\eta^!$,
we similarly have the
DG Lie algebra
\begin{equation}
\bar{\mathfrak D}^{\circ}(A^!)^{\#}:=\Sigma\bar{\mathrm{C}}^\bullet(A^!)\ltimes
\Sigma^{-1-n}\overline{\mathrm{CC}}^\bullet(A^!),
\end{equation}
and
$(0,\Sigma^{-1-n}\eta^!)$ is a Maurer-Cartan element.
However, this is not exactly
the DG Lie algebra that we will discuss in the next section.
In fact, let us first
consider  the Connes cyclic cochain complex
$\mathrm{CC}^\bullet_{\lambda}(A^!)$,
which is a cyclically invariant subcomplex of $\mathrm C^\bullet(A^!)$,
the linear dual of the Hochschild
chain complex of $A$ (recall that it is identified with $\mathrm{C}^\bullet(A; A^*)$).
It is then a direct check that
$
\mathrm{CC}^\bullet_{\lambda}(A^!)
$
is closed under the Lie derivative of
$\bar{\mathrm{C}}^\bullet(A^!)$,
and hence
\begin{equation}\label{DGLAofFrobenius}
\mathfrak D^{\circ}(A^!)^{\#}:=\Sigma\bar{\mathrm{C}}^\bullet(A^!)\ltimes
\Sigma^{-1-n}\mathrm{CC}^\bullet_{\lambda}(A^!)
\end{equation}
is a DG Lie algebra.
Since $\eta^!$ is a cyclically invariant inner product of $A^!$,
$(0,\Sigma^{-1-n}\eta^!)$ is a Maurer-Cartan element of this DG Lie algebra.
Observing that $\overline{\mathrm{CC}}^\bullet(A^!)$ is quasiisomorphic
to the Connes cyclic cochain complex
$\mathrm{CC}^\bullet_{\lambda}(A^!)$ (see Loday \cite[\S2.4]{Loday} for more details),
which is compatible with the Lie derivative actions,
we thus have a quasiisomorphism of DG Lie algebras
$$\bar{\mathfrak D}^{\circ}(A^!)^{\#}\simeq\mathfrak D^{\circ}(A^!)^{\#}.$$
In the following, we write
\begin{equation}\label{DGLAofFrobeniustwisted}
\mathfrak D(A,\eta):=\mathfrak D(A)^{\#}_{(0,\Sigma^{-1-n}\eta)}
\quad\mbox{and}\quad
\mathfrak D^{\circ}(A^!,\eta):=\mathfrak D^{\circ}(A^!)^{\#}_{(0,\Sigma^{-1-n}\eta^!)}.
\end{equation}

\section{Deformation quantization}\label{sect:connections}

In this section, we take $k$ to be a field containing $\mathbb R$.
Dolgushev \cite[Theorem 3]{Dolgushev}
proved that for a Calabi-Yau algebra,
if it is unimodular Poisson, then its
deformation quantization is again Calabi-Yau.
Analogously,
Felder-Shoikhet \cite[Corollary 1]{FS} and
Willwacher-Calaque \cite[Theorem 37]{WC}
proved that for a symmetric Frobenius algebra, if it is
unimodular Frobenius Poisson,
then its deformation quantization is again symmetric Frobenius.
We use their results to
prove Theorems \ref{thm:fifttheorem} and \ref{thm:fourththeorem}.

The following proposition is a rephrase of the results of \S\ref{sect:unimodular_Poisson}
for $A[\![\hbar]\!]$
(see Propositions \ref{prop:altofunimodular}
and \ref{altdefofsymmetricunimodularity}): 

\begin{proposition}\label{altdefofunimodularity}
\textup{(1)} Let $A=k[x_1,\cdots, x_n]$ and $\hbar$ be a formal variable.
For the algebra $A[\![\hbar]\!]$ over $k[\![\hbar]\!]$ together with a bivector
$$
\pi_\hbar:=\hbar\cdot \pi_0+\hbar^2\cdot \pi_1+\cdots\in \hbar\cdot\mathfrak X^{-2}(A[\![\hbar]\!])
$$
and
an $n$-form
$$
\eta_\hbar:=\hbar\cdot \eta_1+\hbar^2\cdot \eta_2+\cdots\in\hbar\cdot\Omega^n(A[\![\hbar]\!]),
$$
the pair $(\pi_\hbar, \eta_0+\eta_{\hbar})$
gives on $A[\![\hbar]\!]$ a unimodular Poisson structure if and only if
$(\Sigma\pi_\hbar, \Sigma^{-1-n}\eta_\hbar)$ is a Maurer-Cartan element
of the DG Lie algebra $\mathfrak P(A[\![\hbar]\!],\eta_0)$.

\textup{(2)} Suppose $A^!=\mathbf\Lambda(\xi_1,\cdots, \xi_n)$ with
volume form $\eta^!_0$.
Then for a bivector $\pi^!_\hbar\in\hbar\cdot\mathfrak X^{-2}(A^![\![\hbar]\!])$
and an $n$-form $\eta^!_\hbar\in\hbar\cdot \mathfrak X^\bullet(A^![\![\hbar]\!]; A^{\ac}[\![\hbar]\!])$,
the pair
$(\pi^!_\hbar, \eta^!_0+\eta^!_\hbar)$ gives an unimodular Frobenius
Poisson structure on $A^![\![\hbar]\!]$
if and only if $(\Sigma\pi^!_\hbar, \Sigma^{-1-n}\eta^!_\hbar)$ is
a Maurer-Cartan element of the DG Lie algebra
$
\mathfrak P^{\circ}(A^![\![\hbar]\!], \eta^!_0).
$
\end{proposition}

For Calabi-Yau algebras and symmetric Frobenius algebras, we have similar results 
(see \eqref{DGLAofCalabiYau}-\eqref{DGLAofFrobeniustwisted}),
due to de Thanhoffer de V\"olcsey-Van den Bergh \cite{dTdVVdB}
and Terilla-Tradler \cite{TerillaTradler} respectively (the interested reader may refer to
these two works for proofs):

\begin{proposition}\label{altdefofCY}
\textup{(1)} \textup{(\cite[Theorem 8.1]{dTdVVdB})}
Suppose $A$ is an $n$-Calabi-Yau algebra with
multiplication $\mu_0$ and volume form $\eta_0$. Then
an element $\mu_\hbar\in \hbar\cdot\bar{\mathrm{C}}^{-2}(A[\![\hbar]\!])$ and
an $n$-form $\eta_\hbar\in\hbar\cdot \overline{\mathrm{CC}}_n^{-}(A[\![\hbar]\!])$
such that $(\mu_0+\mu_\hbar, \eta_0+\eta_\hbar)$ gives a Calabi-Yau structure on $A[\![\hbar]\!]$
if and only if
$
(\Sigma\mu_\hbar,\Sigma^{-1-n}\eta_\hbar)
$
is a Maurer-Cartan element of the DG Lie algebra
$\mathfrak D(A[\![\hbar]\!],\eta_0).$

\textup{(2)} \textup{(\cite[Theorem 3.7]{TerillaTradler})}
Suppose $A^!$ is a symmetric Frobenius algebra with volume $n$-form $\eta^!_0$. Then
an element $\mu^!_\hbar\in \hbar\cdot \bar{\mathrm{C}}^2(A^![\![\hbar]\!])$ and
an $n$-form $\eta^!_\hbar\in\hbar\cdot  {\mathrm{CC}}^n_{\lambda}(A^![\![\hbar]\!])$
such that $(\mu^!_0+\mu^!_\hbar, \eta^!_0+\eta^!_\hbar)$ gives a symmetric Frobenius algebra structure on $A^![\![\hbar]\!]$
if and only if
$
(\Sigma\mu^!_\hbar,\Sigma^{-1-n}\eta^!_\hbar)
$
is a Maurer-Cartan element of the DG Lie algebra
$
\mathfrak D^\circ(A^![\![\hbar]\!],\eta^!_0).
$
\end{proposition}

In fact, in both works, the authors also showed that the DG Lie algebras appeared
in the above proposition
are quasiisomorphic, up to a degree shift, to the negative cyclic chain complex
and the cyclic cochain complex respectively.


\subsection{Deformation quantization of Calabi-Yau Poisson algebras}

In this subsection we prove Theorem \ref{thm:fifttheorem} (1).

Recall that for a Poisson algebra $A$ with bracket $\{-,-\}$,
its {\it deformation quantization}, denoted by $A_{\hbar}$,
is a $k[\![\hbar]\!]$-linear associative product (called the {\it star-product}) on $A[\![\hbar]\!]$
$$
a\ast b=a\cdot b+\mu_1(a,b)\hbar+\mu_2(a,b)\hbar^2+\cdots,
$$
where $\hbar$ is the formal parameter and $\mu_i$ are bilinear operators,
satisfying
$$
\lim_{\hbar\to 0}\frac{1}{\hbar}\left(a\ast b-b\ast a\right)=\{a,b\},\quad\mbox{for all}\; a, b\in A.
$$
In \cite{Kontsevich},
Kontsevich constructed,
for $A$ being the algebra of smooth functions on a Poisson manifold,
an explicit $L_\infty$-quasiisomorphism
from the space of polyvector fields to the Hochschild cochain complex of $A$,
and therefore
there is a one-to-one correspondence between
the equivalence classes of star-products and
the equivalence classes of Poisson algebra structures on
$A[\![\hbar]\!]$.
Thus
via Kontsevich's map, the Poisson bivector $\hbar\pi$ on $A[\![\hbar]\!]$
gives a star-product on $A[\![\hbar]\!]$, which is called {\it Kontsevich's
deformation quantization}.






Note that $\Omega^\bullet(A)$ and $\bar{\mathrm{C}}_\bullet(A)$
are modules over $\mathfrak X^\bullet(A)$ and over $\bar{\mathrm{C}}^\bullet(A)$ respectively,
and in \cite[Conjecture 5.3.2]{Tsygan}, Tsygan conjectured that Kontsevich's deformation
quantization also gives an $L_\infty$-quasiisomorphism of $L_\infty$-modules
between
$
\bar{\mathrm{C}}_\bullet(A)
$
and
$\Omega^\bullet(A)$.
This is known as Tsygan's Formality Conjecture for chains, and is proved
by Shoikhet in \cite[Theorem 1.3.1]{Shoikhet0}.
Shoikhet also conjectured that such $L_\infty$-morphism is also compatible with
the cap product, which was later proved by Calaque and Rossi in \cite[Theorem A]{CR}.

Recall that on $\Omega^\bullet(A)$ and $\bar{\mathrm{C}}_\bullet(A)$, we
have the de Rham differential operator and the Connes boundary operator respectively.
One naturally expects the $L_\infty$-quasiisomorphism constructed
above respects these two operators. This is known as the Cyclic Formality Conjecture
for chains, and is proved by Willwacher in \cite[Theorem 1.3 and Corollary 1.4]{Willwacher}.

With the above results, one obtains the following theorem, due to Dolgushev
\cite[Theorem 3]{Dolgushev} (see also \cite[(1.3)]{dTdVVdB}), whose proof is therefore only sketched:


\begin{theorem}\label{thm:dolgushev}
Let $A=k[x_1,\cdots,x_n]$ be a Poisson algebra.
Then the deformation quantization of $A$ is Calabi-Yau if and only if $A$ is
unimodular.
\end{theorem}

\begin{proof}[Sketch of proof]
Denote by $\mathfrak U$ and $\mathfrak S$ the $L_\infty$-quasiisomorphism
of Kontsevich and of Willwacher respectively, then the works of
\cite{Kontsevich,Willwacher}
are equivalent to saying that there exists
a roof of
$L_\infty$-quasiisomorphisms
$$
\xymatrixcolsep{-20pt}
\xymatrix{
&\Sigma\mathfrak X^\bullet(A[\![\hbar]\!])\ltimes\Sigma^{-1-n}\overline{\mathrm{CC}}_\bullet^{-}(A[\![\hbar]\!])\ar[ld]_{id\times \mathfrak S}\ar[rd]^{\mathfrak U\times id}&\\
\Sigma\mathfrak X^\bullet(A[\![\hbar]\!])\ltimes\Sigma^{-1-n} \Omega^\bullet(A[\![\hbar]\!])[\![u]\!]
&&\Sigma^{}\overline{
\mathrm{C}}^\bullet(A[\![\hbar]\!])\ltimes\Sigma^{-1-n}\overline{\mathrm{CC}}_\bullet^{-}(A[\![\hbar]\!])
} $$
of DG Lie algebras (see \cite[\S11.3]{dTdVVdB} for a proof).

Recall that from Example \ref{exampleofCY} the volume forms in the three DG Lie
modules are the same on the homology level.
Twisting the differentials with the corresponding volume forms
in each of the DG Lie algebra in the above roof we get a new roof
of $L_\infty$-quasiisomorphisms.
This then implies that we have an $L_\infty$-quasiisomorphism of DG Lie algebras
$$
\xymatrix{
\mathfrak P(A[\![\hbar]\!], \eta_0)\ar@{.>}[r]^-{\simeq}&\mathfrak D(A[\![\hbar]\!],\eta_0),
}
$$
where
the dotted arrow means the quasiisomorphism
is given by a sequence of (roofs of) $L_\infty$-morphisms.

As a corollary, the Maurer-Cartan elements
of $\mathfrak P(A[\![\hbar]\!], \eta_0)$
 (up to gauge equivalence)
are in one-to-one correspondence,
via the above $L_\infty$-quasiisomorphisms,
with the Maurer-Cartan elements of
$\mathfrak D(A[\![\hbar]\!],\eta_0)$.
In particular, by Propositions
\ref{altdefofunimodularity} (1) and \ref{altdefofCY} (1),
if $A$ is unimodular Poisson, then $A_{\hbar}$ is Calabi-Yau, and vice versa.
\end{proof}

\begin{proof}[Proof of Theorem \ref{thm:fifttheorem} (1)]
It is proved by Calaque and Rossi in \cite[Theorem 6.1]{CR}
that we have a commutative diagram
\begin{equation}\label{Liemoduleactionsarecompatible}
\xymatrixcolsep{4pc}
\xymatrix{
\mathfrak X^\bullet(A[\![\hbar]\!])\ar@{~>}[r]\ar[d]^{\mathfrak U}_{\simeq}&\Omega^\bullet(A[\![\hbar]\!])\\
\bar{\mathrm{C}}^\bullet(A[\![\hbar]\!])\ar@{~>}[r]
&\bar{\mathrm{C}}_\bullet(A[\![\hbar]\!]),\ar[u]_{\mathfrak S}^{\simeq}
}
\end{equation}
where the horizontal curved arrows mean the cap product.
Since $A$ is unimodular Poisson, $A_\hbar$ is Calabi-Yau,
and $\mathfrak S$ maps the volume form of $A[\![\hbar]\!]$ to the
volume form of $A_\hbar$ under the Hochschild-Kostant-Rosenberg
map, and we thus obtain
the following commutative diagram
$$
\xymatrixcolsep{4pc}
\xymatrix{
\mathrm{HP}^\bullet(A[\![\hbar]\!])\ar[r]^-{\cong}\ar[d]^{\cong}&\mathrm{HP}_{n-\bullet}(A[\![\hbar]\!])
\ar[d]^{\cong}
\\
\mathrm{HH}^\bullet(A_{\hbar})\ar[r]^-{\cong}&\mathrm{HH}_{n-\bullet}(A_{\hbar})}
$$
by Theorem \ref{thm:Xu} and the noncommutative Poincar\'e duality \eqref{thm:dTdVVdB}.
\end{proof}

\subsection{Deformation quantization of Frobenius Poisson algebras}

We first rephrase Kontsevich's Cyclic Formality Conjecture for {\it cochains},
published in
Felder-Shoihket \cite[\S 1]{FS},
in the case $k^{0|n}$.
Note that in this case, the space of functions $\mathscr O(k^{0|n})\cong A^!:=\Lambda^\bullet(\xi_1,\cdots,\xi_n)$.

Recall that by Cattaneo and Felder \cite[Appendix]{CF},
Kontsevich's $L_\infty$-quasiisomorphism holds for the supermanifold case.
Denote this quasiisomorphism again by $\mathfrak U$.
The following is stated in Felder-Shoikhet \cite{FS} and proved by
Willwacher-Calaque \cite[Theorem 2]{WC} (see also
\cite{FS} for some partial results):

\begin{lemma}[Formality for cochains]
For $A^!=\mathscr O(k^{0|n})\cong\Lambda^\bullet(\xi_1,\cdots,\xi_n)$,
there exists an $L_\infty$-quasiisomorphism of Lie modules
$$
\xymatrix{
\mathfrak V: (\mathfrak X^\bullet(A^![\![\hbar]\!]; A^{\ac}[\![\hbar]\!])[\![u]\!], ud^*)
\ar[r]^-{\simeq}&
(\mathrm{CC}^\bullet_{\lambda}(A^![\![\hbar]\!]), \delta).
}
$$
In other words,
there exists an $L_\infty$-quasiisomorphism of Lie algebras:
\begin{equation*}
\xymatrix{
\mathfrak U\times\mathfrak V:
\Sigma\mathfrak X^\bullet(A^![\![\hbar]\!])\ltimes 
\Sigma^{-1-n}\mathfrak X^\bullet(A^![\![\hbar]\!];A^{\ac}[\![\hbar]\!])[\![u]\!]\ar[r]^-{\simeq}
&
\Sigma\bar{\mathrm{C}}^\bullet(A^![\![\hbar]\!])\ltimes\Sigma^{-1-n}\mathrm{CC}^\bullet_{\lambda}(A^![\![\hbar]\!]).}
\end{equation*}
\end{lemma}

Again we recommend \cite[\S11]{dTdVVdB} for the formulas of the (Taylor) expansion
of $\mathfrak U\times\mathfrak V$.
Also, we mention that the first term of the above $L_\infty$-quasiisomorphism
$\mathfrak V$ is the Hochschild-Kostant-Rosenberg map, 
which then preserves the volume forms on each side.
Therefore, we get a quasiisomorphism
$$\mathfrak P^\circ(A^![\![\hbar]\!], \eta^!_0)
\simeq\mathfrak D^\circ(A^![\![\hbar]\!], \eta^!_0)$$ as DG Lie algebras.
As a corollary, we have the following theorem, due to
Felder-Shoikhet \cite[Corollary 1]{FS} and Willwacher-Calaque \cite[Theorem 37]{WC}:

\begin{theorem}\label{thm:FSWC}
For $A^!=\Lambda^\bullet(\xi_1,\cdots,\xi_n)$, the deformation quantization
of $A^!$ is symmetric Frobenius if and only if $A^!$ is unimodular Frobenius.
\end{theorem}

\begin{proof}[Proof of Theorem \ref{thm:fifttheorem} (2)]
Recall that
$\Omega^\bullet(A^![\![\hbar]\!])$ and $\bar{\mathrm C}_\bullet(A^![\![\hbar]\!])$
are Lie modules over $\mathfrak X^\bullet(A^![\![\hbar]\!])$ 
and $\bar{\mathrm C}^\bullet(A^![\![\hbar]\!])$
respectively.
Applying Calaque-Rossi's result \eqref{Liemoduleactionsarecompatible} to $A^![\![\hbar]\!]$, 
we have
the following commutative diagram
\begin{equation*}
\xymatrixcolsep{4pc}
\xymatrix{
\mathfrak X^\bullet(A^![\![\hbar]\!])\ar@{~>}[r]\ar[d]^{\mathfrak U}_{\simeq}
&\Omega^\bullet(A^![\![\hbar]\!])\\
\bar{\mathrm{C}}^\bullet(A^![\![\hbar]\!])\ar@{~>}[r]
&\bar{\mathrm{C}}_\bullet(A^![\![\hbar]\!]).\ar[u]_{\mathfrak S}^{\simeq}
}
\end{equation*}
Now consider the adjoint actions of the Lie algebras to the linear dual spaces of the Lie modules 
(see Remark \ref{Gerstenhabercomodule}),
we obtain the following commutative diagram
\begin{equation*}
\xymatrixcolsep{4pc}
\xymatrix{
\mathfrak X^\bullet(A^![\![\hbar]\!])\ar@{~>}[r]\ar[d]^{\mathfrak U}_{\simeq}
&\mathfrak X^\bullet(A^![\![\hbar]\!]; A^{\ac}[\![\hbar]\!])
\ar[d]^{\mathfrak S^*}_{\simeq}\\
\bar{\mathrm{C}}^\bullet(A^![\![\hbar]\!])\ar@{~>}[r]
&\bar{\mathrm{C}}^\bullet(A^![\![\hbar]\!]; A^{\ac}[\![\hbar]\!]).
}
\end{equation*}
Taking the homology in the above commutative diagram and applying the 
Poincar\'e duality, whose existence is guaranteed by Theorem \ref{thm:FSWC}, we obtain
the commutative diagram
$$
\xymatrixcolsep{4pc}
\xymatrix{\mathrm{HP}^\bullet(A^![\![\hbar]\!])\ar[r]^-{\cong}\ar[d]^{\cong}
&\mathrm{HP}^{\bullet-n}(A^![\![\hbar]\!]; A^{\ac}[\![\hbar]\!])\ar[d]^{\cong}\\
\mathrm{HH}^\bullet(A^!_{\hbar})\ar[r]^-{\cong}
&\mathrm{HH}^{\bullet-n}(A^!_{\hbar}; A^{\ac}_{\hbar}).
}
$$
This completes the proof.
\end{proof}


\begin{proof}[Proof of Theorem \ref{thm:fourththeorem}]
By Shoikhet \cite[Theorem 0.3]{Shoikhet}
(see also \cite[Theorem 8.6]{CFFR}), $A_{\hbar}$ and $A^{!}_{\hbar}$ are Koszul dual algebras over $k[\![\hbar]\!]$,
and hence the theorem follows from a combination of Theorems
\ref{thm:thirdtheorem},
\ref{thm:fifttheorem},
and \ref{thm:conjofRouquier}.
\end{proof}

\subsection{Twisted Poincar\'e duality for Poisson algebras}

For a general associative algebra, say $A$,
it may not be Calabi-Yau, and therefore there may not exist any
Poincar\'e duality between $\mathrm{HH}^\bullet(A)$ and $\mathrm{HH}_\bullet(A)$.
In \cite{BZ},
Brown and Zhang
introduced the so-called ``twisted Poincar\'e duality" for
associative algebras. That is, for such $A$,
keeping its left $A$-module structure (the multiplication)
as usual, the right $A$-module structure of $A$ is the multiplication
composed with an automorphism $\sigma:A\to A$.
Denote such $A$-bimodule by $A_{\sigma}$, then
Brown and Zhang showed that for a lot of algebras,
there exists a twisted Poincar\'e duality
$
\mathrm{HH}^\bullet(A)\cong\mathrm{HH}_{n-\bullet}(A; A_{\sigma})
$
for some $n\in\mathbb N$ (cf. \cite[Corollary 5.2]{BZ}).
In this case $A$ is called a {\it twisted Calabi-Yau} algebra of dimension $n$.

Such phenomenon also occurs for Poisson algebras. Namely, not
all Poisson algebras are unimodular, and hence there may not exist
an isomorphism between
$
\mathrm{HP}^\bullet(A)
$
and
$\mathrm{HP}_{\bullet}(A)$.
In \cite{LR07,LWW,ZVOZ,Zhu},
the authors studied the so-called twisted Poincar\'e duality for Poisson algebras,
similarly to that of associative algebras.
They also studied some comparisons with twisted Calabi-Yau algebras.
However, it would be very interesting
to study the relationships between the
deformation quantization of twisted unimodular Poisson algebras
and twisted Calabi-Yau algebras, and obtain a theorem similar to
Theorem \ref{thm:fourththeorem}
in this twisted case.


\end{document}